\documentclass[11pt,letterpaper]{article}
\usepackage[margin=1in]{geometry}

\usepackage{thmtools}
\usepackage{thm-restate}
\usepackage{bm}
\usepackage{mathrsfs}           %
\usepackage{tikz}\usetikzlibrary{positioning,chains,fit,shapes,calc}

\usepackage{subcaption}
\usepackage{algorithm}
\usepackage{algorithmic}
\usepackage[ruled,linesnumbered,algo2e]{algorithm2e}
\usepackage{enumerate}

\usepackage{amsmath,amssymb, amsthm}    %
\usepackage{verbatim}           %
\usepackage{xspace}             %
\usepackage{graphicx,float}     %
\usepackage{ifthen,calc}        %
\usepackage{textcomp}           %
\usepackage{fancybox}           %
\usepackage{hhline}             %
\usepackage{float}              %

\usepackage[vflt]{floatflt}     %
\usepackage[small,compact]{titlesec}
\usepackage{setspace}
\usepackage{color}
\definecolor{ForestGreen}{rgb}{0.1333,0.5451,0.1333}
\usepackage{thumbpdf}
\usepackage[letterpaper,
colorlinks,linkcolor=ForestGreen,citecolor=ForestGreen,
backref,
bookmarks,bookmarksopen,bookmarksnumbered]
{hyperref}

\newtheorem{theorem}{Theorem}
\newtheorem{claim}{Claim}
\newtheorem{proposition}{Proposition}

\newtheorem{definition}{Definition}

\newtheorem{remark}{Remark}
\newtheorem{lemma}{Lemma}
\newtheorem{fact}{Fact}

\usepackage{float}
\floatstyle{plain}\newfloat{myfig}{t}{figs}[section]
\floatname{myfig}{\textsc{Figure}}
\floatstyle{plain}\newfloat{myalg}{H}{algs}[section]
\floatname{myalg}{}
\newcommand{\defeq}{:=}
\newcommand{\norm}[1]{\left\lVert#1\right\rVert}
\newcommand{\inprod}[2]{\left\langle#1, #2\right\rangle}

\newcommand{\argmin}{\textup{argmin}} 
\newcommand{\R}{\mathbb{R}}

\newcommand{\N}{\mathbb{N}}
\newcommand{\diag}[1]{\textbf{\textup{diag}}\left(#1\right)}
\newcommand{\half}{\frac{1}{2}}

\newcommand{\E}{\mathbb{E}}

\newcommand{\Nor}{\mathcal{N}}
\newcommand{\smax}{\textup{smax}}
\newcommand{\sqmax}{\textup{sqmax}}

\newcommand{\xset}{\mathcal{X}}

\newcommand{\ma}{\mathbf{A}}

\newcommand{\ai}{\ma_{i:}}
\newcommand{\aj}{\ma_{:j}}
\newcommand{\id}{\mathbf{I}}
\newcommand{\tO}{\widetilde{O}}

\newcommand{\mw}{\mathbf{W}}
\newcommand{\Par}[1]{\left(#1\right)}
\newcommand{\Brack}[1]{\left[#1\right]}

\newcommand{\Abs}[1]{\left|#1\right|}

\newcommand{\mzero}{\mathbf{0}}

\newcommand{\bw}{\bar{w}}

\newcommand{\msig}{\boldsymbol{\Sigma}}
\newcommand{\Sym}{\mathbb{S}}
\newcommand{\PSD}{\Sym_{\succeq 0}}
\newcommand{\PD}{\Sym_{\succ 0}}
\newcommand{\sx}{x^\star}
\newcommand{\pstar}{p^\star}
\newcommand{\estar}{e^\star}
\newcommand{\mg}{\mathbf{G}}

\newcommand{\xin}{x_{\textup{in}}}
\newcommand{\xout}{x_{\textup{out}}}

\newcommand{\wt}{\widetilde}

\newcommand{\tp}{\tilde{p}}
\newcommand{\zalg}{\zeta_{\text{alg}}}
\newcommand{\NS}{\textup{NS}}
\newcommand{\Cp}{C_{\textup{prog}}}
\newcommand{\HRS}{\mathsf{HalfRadiusSparse}}
\newcommand{\HRSN}{\mathsf{HalfRadiusSparseNoisy}}
\newcommand{\OStep}{\mathsf{StepOracle}}
\newcommand{\ostep}{\mathcal{O}_{\textup{step}}}
\newcommand{\onstep}{\mathcal{O}_{\textup{nstep}}}
\newcommand{\ONStep}{\mathsf{StrongStepOracle}}
\newcommand{\tDelta}{\widetilde{\Delta}}

\newcommand{\trunc}{\textup{trunc}}
\newcommand{\mproj}{\boldsymbol{\Pi}}
\newcommand{\wsinf}{w^\star_\infty}

\newcommand{\Apot}{\Phi_{2}}
\newcommand{\Bpot}{\Phi_{\sqmax}}
\newcommand{\tApot}{\widetilde{\Phi}_{2}}
\newcommand{\tBpot}{\widetilde{\Phi}_{\sqmax}}

\newcommand{\codeInput}{\textbf{Input:} }
\newcommand{\codeOutput}{\textbf{Output:} }
\newcommand{\codeReturn}{\textbf{Return:} }
\newcommand{\codeBreak}{\textbf{Break:} }
\newcommand{\codeLineSace}{\BlankLine}

\definecolor{burntorange}{rgb}{0.8, 0.33, 0.0}

  \usepackage{nth}
  \usepackage{intcalc}

\usepackage{url}

\begin{document}

	\begin{titlepage}
		\def\thepage{}
		\thispagestyle{empty}
		
		\title{Semi-Random Sparse Recovery in Nearly-Linear Time} 
		
		\date{}
		\author{
			Jonathan A.\ Kelner\thanks{MIT, {\tt kelner@mit.edu}.  Supported in part by NSF awards CCF-1955217,  CCF-1565235, and DMS-2022448.}
			\and
			Jerry Li\thanks{Microsoft Research, {\tt jerrl@microsoft.com}. This work was partially done while visiting the Simons Institute for the Theory of Computing.}
			\and
			Allen Liu\thanks{MIT, {\tt cliu568@mit.edu}. This work was partially done while working as an intern at Microsoft Research, and was supported in part by an NSF Graduate Research Fellowship and a Fannie and John Hertz Foundation Fellowship.}
			\and
			Aaron Sidford\thanks{Stanford University, {\tt sidford@stanford.edu}. Supported in part by a Microsoft Research Faculty Fellowship, NSF CAREER Award CCF-1844855, NSF Grant CCF-1955039, a PayPal research award, and a Sloan Research Fellowship.}
			\and
			Kevin Tian\thanks{Stanford University, {\tt kjtian@stanford.edu}. This work was partially done while visiting the Simons Institute for the Theory of Computing, and was supported in part by a Google Ph.D.\ Fellowship, a Simons-Berkeley VMware Research Fellowship, a Microsoft Research Faculty Fellowship, NSF CAREER Award CCF-1844855, NSF Grant CCF-1955039, and a PayPal research award.}
		}
		
		\maketitle

\abstract{
Sparse recovery is one of the most fundamental and well-studied inverse problems.
Standard statistical formulations of the problem are provably solved by general convex programming techniques and more practical, fast (nearly-linear time) iterative methods. However, these latter ``fast algorithms'' have previously been observed to be brittle in various real-world settings.

We investigate the brittleness of fast sparse recovery algorithms to generative model changes through the lens of studying their robustness to a ``helpful'' semi-random adversary, a framework which tests whether an algorithm overfits to input assumptions. We consider the following basic model: let $\ma \in \R^{n \times d}$ be a measurement matrix which contains an \emph{unknown} subset of rows $\mg \in \R^{m \times d}$ which are bounded and satisfy the restricted isometry property (RIP), but is otherwise arbitrary. Letting $\sx \in \R^d$ be $s$-sparse, and given either exact measurements $b = \ma \sx$ or noisy measurements $b = \ma \sx + \xi$, we design algorithms recovering $\sx$ information-theoretically optimally in nearly-linear time. We extend our algorithm to hold for weaker generative models relaxing our planted RIP row subset assumption to a natural weighted variant, and show that our method's guarantees naturally interpolate the quality of the measurement matrix to, in some parameter regimes, run in sublinear time.

Our approach differs from that of prior fast iterative methods with provable guarantees under semi-random generative models \cite{cheng2018non, li2020wellconditioned}, which typically separate the problem of learning the planted instance from the estimation problem, i.e.\ they attempt to first learn the planted ``good'' instance (in our case, the matrix $\mg$). However, natural conditions on a submatrix which make sparse recovery tractable, such as RIP, are NP-hard to verify and hence first learning a sufficient row reweighting appears challenging. We eschew this approach and design a new iterative method, tailored to the geometry of sparse recovery, which is provably robust to our semi-random model. Our hope is that our approach opens the door to new robust, efficient algorithms for other natural statistical inverse problems.
}
 		
	\end{titlepage}
	\pagenumbering{gobble}
	\setcounter{tocdepth}{2}
	{
		\hypersetup{linkcolor=black}
		\tableofcontents
	}
	\newpage
	\pagenumbering{arabic}

\section{Introduction}
\label{sec:intro}

Sparse recovery is one of the most fundamental and well-studied inverse problems, with numerous applications in prevalent real-world settings \cite{eldar2012compressed}. In its most basic form, we are given an entrywise Gaussian measurement matrix $\mg \in \R^{m \times d}$ and measurements $b = \mg \sx$ for an unknown $s$-sparse $\sx \in \R^d$; the goal of the problem is to recover $\sx$. 
Seminal works by Cand\`{e}s, Romberg, and Tao~\cite{candes2005signal,candes2006near,candes2006stable} showed that even when the linear system in $\mg$ is extremely underconstrained, recovery is tractable so long as $m = \Omega (s \log d)$. Further they gave a polynomial-time algorithm known as \emph{basis pursuit} based on linear programming recovering $\sx$ in this regime.

Unfortunately, the runtime of linear programming solvers, while polynomial in the size of the input, can still be prohibitive in many high-dimensional real-world settings. Correspondingly, a number of alternative approaches which may broadly be considered first-order methods have been developed. These methods provably achieve similar recovery guarantees under standard generative models such as Gaussian measurements, with improved runtimes compared to the aforementioned convex programming methods. We refer to these first-order methods through as ``fast'' algorithms throughout and they may roughly be placed in the following (potentially non-disjoint) categories.
\begin{itemize}
	\item \textbf{Greedy algorithms}, e.g.~\cite{mallat1993matching,pati1993orthogonal,needell2010signal}, seek to greedily find elements in the support of the true $\sx$ using different combinatorial search criteria. 
	\item \textbf{Non-convex iterative algorithms}, e.g.~\cite{needell2009cosamp,blumensath2009iterative,blumensath2010normalized,maleki2010optimally,foucart2011hard}, directly optimize a (potentially non-convex) objective over a non-convex domain. 
	\item \textbf{Convex first-order methods}, e.g.~\cite{figueiredo2003algorithm,daubechies2004iterative,combettes2005signal,beck2009fast,becker2011nesta,negahban2012unified,agarwal2012fast} quickly solve the convex objective underlying basis pursuit using first-order methods. 
\end{itemize}

We also note that theoretically, when $n$ is sufficiently large, recent advances by \cite{BrandLSS20, BrandLLSS0W21}, also obtain fast runtimes for the relevant linear programming objective. The fastest IPM for the noiseless sparse recovery objective runs in time\footnote{We use $\tO$ to hide polylogarithmic factors in problem parameters for brevity of exposition throughout the paper.} $\tO(nd + n^{2.5})$ which is nearly-linear when $\ma$ is dense and $n \ll d^{2/3}$. For a range of (superlogarithmic, but sublinear) $n$, these runtimes are no longer nearly-linear; furthermore, these IPMs are second-order and our focus is on designing first-order sparse recovery methods, which are potentially more practical.\footnote{We also note that these IPM results do not immediately apply to natural (nonlinear) convex programs for sparse recovery under noisy observations, see Appendix~\ref{app:failure}.}

It has often been observed empirically that fast first-order methods can have large error, or fail to converge, in real-world settings~\cite{davenport2013signal,jain2014iterative,polania2014exploiting} where convex programming-based algorithms (while potentially computationally cumbersome) perform well statistically~\cite{zhang2016comparison,aich2017application}. This may be surprising, given that in theory, fast algorithms essentially match the statistical performance of the convex programming-based algorithms under standard generative assumptions.
While there have been many proposed explanations for this behavior, one compelling argument is that fast iterative methods used in practice are more brittle to changes in modeling assumptions.
We adopt this viewpoint in this paper, and develop fast sparse recovery algorithms which achieve optimal statistical rates under a \emph{semi-random adversarial model}~\cite{blum1995coloring,feige2001heuristics}, a popular framework for investigating the robustness of learning algorithms under changes to the data distribution.

\paragraph{Semi-random adversaries.} 
Semi-random adversaries are a framework for reasoning about algorithmic robustness to distributional shift. They are defined in statistical settings, and one common type of semi-random adversary is one which corresponds to generative models where data has been corrupted in a ``helpful'' or ``monotone'' way. Such a monotone semi-random adversary takes a dataset from which learning is information-theoretically tractable, and augments it with additional information; this additional information may not break the problem more challenging from an information-theoretic perspective,\footnote{There are notable exceptions, e.g.\ the semi-random stochastic block model of \cite{moitra2016robust}.} but may affect the performance of algorithms in other ways. In this paper, we consider a semi-random adversary which makes the \emph{computational problem} more difficult without affecting the problem information-theoretically, by returning a consistent superset of the unaugmented observations.
This contrasts with other adversarial models such as gross corruption~\cite{anscombe1960rejection,tukey1960survey,huber1964robust,tukey1975mathematics}, where corruptions may be arbitrary, and the corrupted measurements incorrect.
It may be surprising that a ``helpful'' adversary has any implications whatsoever on a learning problem, from either an information-theoretic or computational standpoint.

Typically, convex programming methods for statistical recovery problems are robust to these sorts of perturbations --- in brief, this is because constraints to a convex program that are met by an optimum point does not change the optimality of that point.
However, greedy and non-convex methods --- such as popular practical algorithms for sparse linear regression --- can be susceptible to semi-random adversaries.
Variants of this phenomenon have been reported in many common statistical estimation problems, such as stochastic block models and broadcast tree models~\cite{moitra2016robust}, PAC learning~\cite{blum2003machine}, matrix completion~\cite{moitra2017robustness,cheng2018non}, and principal component regression~\cite{bhaskara2021principal}.
This can be quite troubling, as semi-random noise can be thought of as a relatively mild form of generative model misspecification: in practice, the true distribution is almost always different from the models considered in theory.
Consequently, an algorithm's non-robustness to semi-random noise is suggestive that the algorithm may be more unreliable in real-world settings.

We consider a natural semi-random adversarial model for sparse recovery (see e.g.\ page 284 of \cite{AwasthiV18}), which extends the standard restricted isometry property (RIP) assumption, which states that applying matrix $\ma$ approximately preserves the $\ell_2$ norm of sparse vectors. Concretely, throughout the paper we say matrix $\ma$ satisfies the \emph{$(s, c)$-restricted isometry (RIP) property} if for all $s$-sparse vectors $v$,
\[\frac 1 c \norm{v}_2^2 \le \norm{\ma v}_2^2 \le c \norm{v}_2^2.\]
We state a basic version of our adversarial model here, and defer the statement of the fully general version to Definition~\ref{def:wrip}.\footnote{When clear from context, as it will be throughout the main sections of the paper, $s$ will always refer to the sparsity of a vector $\sx \in \R^d$ in an exact or noisy recovery problem through $\ma \in \R^{n \times d}$. For example, the parameter $s$ in Definition~\ref{def:prip} is the sparsity of the vector in an associated sparse recovery problem.} We defer the introduction of notation used in the paper to Section~\ref{sec:prelims}.

\begin{definition}[pRIP matrix]\label{def:prip}
Let $m, n, d \in \N$ be known with $n \geq m$. We say $\ma \in \R^{n \times d}$ is $\rho$-pRIP (planted RIP) if there is an (unknown) $\mg \in \R^{m \times d}$ such that each row of $\mg$ is also a row of $\ma$ and $\frac 1 {\sqrt m} \mg$ is $(\Theta(s), \Theta(1))$-RIP for appropriate constants, and $\norm{\mg}_{\max} \le \rho$. When $\rho = \tO(1)$ for brevity we say $\ma$ is pRIP.
\end{definition}

Under the problem parameterizations used in this paper, standard RIP matrix constructions satisfy $\rho = \tO(1)$ with high probability. For example, when $\mg$ is entrywise Gaussian and $m = \Theta(s\log d)$, a tail bound shows that with high probability we may set $\rho = O(\sqrt{\log d})$ to be compatible with the assumptions in Definition~\ref{def:prip}. 

pRIP matrices can naturally be thought of as arising from a semi-random adversarial model as follows. First, an RIP matrix $\mg \in \R^{m \times d}$ is generated, for example from a standard ensemble (e.g.\ Gaussian or subsampled Hadamard). An adversary inspects $\mg$, and forms $\ma \in \R^{n \times d}$ by reshuffling and arbitrarily augmenting rows of $\mg$. Whenever we refer to a ``semi-random adversary'' in the remainder of the introduction, we mean the adversary provides us a pRIP measurement matrix $\ma$.

The key recovery problem we consider in this paper is recovering an unknown $s$-sparse vector $\sx \in \R^d$ given measurements $b \in \R^n$ through $\ma$. We consider both the \emph{noiseless} or \emph{exact} setting where $b = \ma \sx$ and the \emph{noisy} setting where $b = \ma \sx + \xi$ for bounded $\xi$.
In the noiseless setting in particular, the semi-random adversary hence only gives the algorithm additional \emph{consistent} measurements of the unknown $s$-sparse vector $\sx$. In this sense, the adversary is only ``helpful,'' as it returns a superset of information which is sufficient for sparse recovery (formally, this adversary cannot break the standard restricted nullspace condition which underlies the successful performance of convex programming methods).
We note $n$ may be much larger than $m$, i.e.\ we impose no constraint on how many measurements the adversary adds. 

\paragraph{Semi-random sparse recovery in nearly-linear time.} 
We devise algorithms which match the nearly-linear runtimes and optimal recovery guarantees of faster algorithms on fully random data, but which retain both their runtime and the robust statistical performances of convex programming methods against semi-random adversaries. In this sense, our algorithms obtain the ``best of both worlds.'' We discuss and compare more extensively to existing sparse recovery algorithms under Definition~\ref{def:prip} in the following section.
We first state our result in the noiseless observation setting.
\begin{theorem}[informal, see Theorem~\ref{thm:exact}]\label{thm:exactintro}
	Let $\sx \in \R^d$ be an unknown $s$-sparse vector.
	Let $\ma \in \R^{n \times d}$ be pRIP.
	There is an algorithm, which given $\ma$ and $b = \ma \sx$, runs in time $\widetilde{O} (n d)$, and outputs $\sx$ with high probability.
\end{theorem}
Since our problem input is of size $nd$, our runtime in Theorem~\ref{thm:exactintro} is nearly-linear in the problem size. We also extend our algorithm to handle the noisy observation setting, where we are given perturbed linear measurements of $\sx$ from a pRIP matrix.
\begin{theorem}[informal, see Theorem~\ref{thm:noisy}]\label{thm:noisyintro}
	Let $\sx \in \R^d$ be an unknown $s$-sparse vector, and let $\xi \in \R^n$ be arbitrary.
	Let $\ma \in \R^{n \times d}$ be pRIP.
	There is an algorithm, which given $\ma$ and $b = \ma \sx + \xi$, runs in time $\widetilde{O} (n d)$, and with high probability outputs $x$ satisfying
	\[
	\norm{x - \sx}_2 \leq O \left( \frac{1}{\sqrt{m}} \norm{\xi}_{2, (m)} \right) \; ,
	\]
	where $\norm{\xi}_{2, (m)}$ denotes the $\ell_2$ norm of the largest $m$ entries of $\xi$ by absolute value.
\end{theorem}
The error scaling of Theorem~\ref{thm:noisyintro} is optimal in the semi-random setting.
Indeed, when there is no semi-random noise, the guarantees of Theorem~\ref{thm:noisyintro} exactly match the standard statistical guarantees in the fully-random setting for sparse recovery, up to constants; for example, when $\ma = \sqrt{m} \id$ (which is clearly RIP, in fact an exact isometry, after rescaling), it is information-theoretically impossible to obtain a better $\ell_2$ error.\footnote{In the literature it is often standard to scale down the sensing matrix $\ma$ by $\sqrt{m}$; this is why our error bound is similarly scaled. However, this scaling is more convenient for our analysis, especially when stating weighted results.} The error bound of Theorem~\ref{thm:noisyintro} is similarly optimal in the semi-random setting because in the worst case, the largest entries of $\xi$ may correspond to the rows of the RIP matrix from which recovery is information-theoretically possible.

\paragraph{Performance of existing algorithms.} To contextualize Theorems~\ref{thm:exactintro} and~\ref{thm:noisyintro}, we discuss the performance of existing algorithms for sparse recovery under the semi-random adversarial model of Definition~\ref{def:prip}. First, it can be easily verified that our semi-random adversary never changes the information-theoretic tractability of sparse recovery.
In the noiseless setting for example, the performance of the minimizer to the classical convex program based on $\ell_1$ minimization,
\[\min_{\ma x = b} \norm{x}_1,\] 
is unchanged in the presence of pRIP matrices (as $\sx$ is still consistent with the constraint set, and in particular a RIP constraint set), and hence the semi-random problem can be solved in polynomial time via convex programming. This suggests the main question we address: can we design a near-linear time algorithm obtaining optimal statistical guarantees under pRIP measurements?

As alluded to previously, standard greedy and non-convex methods we have discussed may fail to converge to the true solution against appropriate semi-random adversaries generating pRIP matrices. We give explicit counterexamples to several popular methods such as orthogonal matching pursuit and iterative hard thresholding in Appendix~\ref{app:failure}.
Further, it seems likely that similar counterexamples also break other, more complex methods commonly used in practice, such as matching pursuit~\cite{mallat1993matching} and CoSaMP \cite{needell2009cosamp}.

Additionally, while fast ``convex'' iterative algorithms (e.g.\ first-order methods for solving objectives underlying polynomial-time convex programming approaches) will never fail to converge to the correct solution given pRIP measurements, the analyses which yield fast runtimes for these algorithms~\cite{negahban2012unified,agarwal2012fast} rely on properties such as restricted smoothness and strong convexity (a specialization of standard conditioning assumptions to numerically sparse vectors). These hold under standard generative models but again can be broken by pRIP measurements; consequently, standard convergence analyses of ``convex'' first-order methods may yield arbitrarily poor rates.

One intuitive explanation for why faster methods fail is that they depend on conditions such as incoherence~\cite{donoho1989uncertainty} or the restricted isometry property~\cite{candes2006near}, which can be destroyed by a semi-random adversary.
For instance, RIP states that if $S$ is any subset of $m = \Theta (s)$ columns of $\ma$, and $\ma_S$ is the submatrix formed by taking those columns of $\ma$, then $\ma_S^\top \ma_S$ is an approximate isometry (i.e.\ it is well-conditioned).
While it is well-known that RIP is satisfied with high probability when $\ma$ consists of $\Theta(s \log d)$ Gaussian rows, it is not too hard to see that augmenting $\ma$ with additional rows can easily ruin the condition number of submatrices of this form. 
In contrast, convex methods work under weaker assumptions such as the restricted nullspace condition, which cannot be destroyed by the augmentation used by pRIP matrices. Though these weaker conditions (e.g.\ the restricted nullspace condition) suffice for algorithms based on convex programming primitives, known analyses of near-linear time ``fast'' algorithms require additional instance structure, such as incoherence or RIP.
Thus, it is plausible that fast algorithms for sparse recovery are less robust to the sorts of distributional changes that may occur in practice.

\paragraph{Beyond submatrices.} Our methods naturally extend to a more general setting (see Definition~\ref{def:wrip}, wherein we define ``weighted RIP'' (wRIP) matrices, a generalization of Definition~\ref{def:prip}).
Rather than assuming there is a RIP submatrix $\mg$, we only assume that there is a (nonnegative) reweighting of the rows of $\ma$ so that the reweighted matrix is ``nice,'' i.e.\ it satisfies RIP.
Definition~\ref{def:prip} corresponds to the special case of this assumption where the weights are constrained to be either $0$ or $1$ (and hence must indicate a subset of rows).
In our technical sections (Sections~\ref{sec:exact} and~\ref{sec:noisy}), our results are stated for this more general semi-random model, i.e.\ sparse recovery from wRIP measurements.
For simplicity of exposition, throughout the introduction, we mainly focus on the simpler pRIP sparse recovery setting described following Definition~\ref{def:prip}.

\paragraph{Towards instance-optimal guarantees.} While the performance of the algorithms in Theorems~\ref{thm:exactintro} and~\ref{thm:noisyintro} is already nearly-optimal in the worst case semi-random setting, one can still hope to improve our runtime and error bounds in certain scenarios. Our formal results, Theorems~\ref{thm:exact} and~\ref{thm:noisy}, provide these types of fine-grained instance-optimal guarantees in several senses.

In the noiseless setting (Theorem~\ref{thm:exact}), if it happens to be that the entire matrix $\ma$ is RIP (and not just $\mg$), then standard techniques based on subsampling the matrix can be used to solve the problem in time $\widetilde{O}(s d)$ with high probability. For example, if $\ma$ is pRIP where, following the notation of Definition~\ref{def:prip}, $\mg$ is entrywise Gaussian, and the adversary chose to give us additional Gaussian rows, one could hope for a runtime improvement (simply by ignoring the additional measurements given). 
Theorem~\ref{thm:exact} obtains a runtime which smoothly interpolates between the two regimes of a worst-case adversary and an adversary which gives us additional random measurements from an RIP ensemble. Roughly speaking, if there exists a (a priori unknown) submatrix of $\ma$ of $m \gg \widetilde{\Theta}(s)$ rows which is RIP, then we show that our algorithm runs in \emph{sublinear} time $\widetilde{O}(nd \cdot \frac s m)$, which is $\tO(sd)$ when $m\approx n$.
We show this holds in our weighted semi-random model (under wRIP measurements, Definition~\ref{def:wrip}) as well, where the runtime depends on the ratio of the $\ell_1$ norm of the (best) weight vector to its $\ell_\infty$ norm, a continuous proxy for the number of RIP rows under pRIP.

We show a similar interpolation holds in the noisy measurement setting, both in the runtime sense discussed previously, and also in a statistical sense. In particular, Theorem~\ref{thm:noisy} achieves (up to logarithmic factors) the same interpolating runtime guarantee of Theorem~\ref{thm:exact}, but further attains a squared $\ell_2$ error which is roughly the average of the $m$ largest elements of the squared noise vector $\xi$ (see the informal statement in Theorem~\ref{thm:noisyintro}). This bound thus improves as $m \gg \widetilde{\Theta}(s)$; we show it extends to weighted RIP matrices (Definition~\ref{def:wrip}, our generalization of Definition~\ref{def:prip}) in a natural way depending on the $\ell_\infty$-$\ell_1$ ratio of the weights.

\subsection{Our techniques}

Our overall approach for semi-random sparse recovery is fairly different from two recent works in the literature which designed fast iterative methods succeeding under a semi-random adversarial model \cite{cheng2018non, li2020wellconditioned}. In particular, these two algorithms were both based on the following natural framework, which separates the ``planted learning'' problem (e.g.\ identifying the planted benign matrix) from the ``estimation'' task (e.g.\ solving a linear system or regression problem).
\begin{enumerate}
    \item Compute a set of weights for the data (in linear regression for example, these are weights on each of the rows of a measurement matrix $\ma$), such that after re-weighting, the data fits the input assumptions of a fast iterative method which performs well on a fully random instance.
    \item Apply said fast iterative algorithm on the reweighted data in a black-box manner.
\end{enumerate} 
To give a concrete example, \cite{li2020wellconditioned} studied the standard problem of \emph{overdetermined} linear regression with a semi-random adversary, where a measurement matrix $\ma$ is received with the promise that $\ma$ contains a ``well-conditioned core'' $\mg$. The algorithm of \cite{li2020wellconditioned} first learned a re-weighting of the rows of $\ma$ by a diagonal matrix $\mw^{\half}$, such that the resulting system in $\ma^\top \mw \ma$ is well-conditioned and hence can be solved using standard first-order methods.

In the case of semi-random sparse recovery, there appear to be significant barriers to reweighting approaches (which we will shortly elaborate on).  We take a novel direction that involves designing a new nearly-linear time iterative method for sparse recovery tailored to the geometry of the problem. 

\paragraph{Why not reweight the rows?} There are several difficulties which are immediately encountered when one tries to use the aforementioned reweighting framework for sparse recovery. First of all, there is no effective analog of condition number for an underdetermined linear system. The standard assumption on the measurement matrix $\ma$ to make sparse recovery tractable for fast iterative methods is that $\ma$ satisfies RIP, i.e.\ $\ma$ is roughly an isometry when restricted to $O(s)$-sparse vectors. However, RIP is NP-hard to verify~\cite{bandeira2013certifying} and this may suggest that it is computationally hard to try, say, learning a reweighting of the rows of $\ma$ such that the resulting reweighted matrix is guaranteed to be RIP (though it would be very interesting if this were achievable). More broadly, almost all explicit conditions (e.g.\ RIP, incoherence etc.) which make sparse recovery tractable for fast algorithms are conditions about subsets of the \emph{columns} of $\ma$.  Thus, any approach which reweights rows of $\ma$ such that column subsets of the reweighted matrix satisfy an appropriate condition results in optimization problems that seems challenging to solve in nearly-linear time.

\paragraph{The geometry of sparse recovery.} We now explain our new approach, and how we derive deterministic conditions on the steps of an iterative method which certify progress by exploiting the geometry of sparse recovery. We focus on the clean observation setting in this technical overview. Suppose that we wish to solve a sparse regression problem $\ma \sx = b$ where $\sx$ is $s$-sparse, and we are given $\ma$ and $b$. To fix a scale, suppose for simplicity that we know $\norm{\sx}_1 = \sqrt s$ and $\norm{\sx}_2 = 1$. Also, assume for the purpose of conveying intuition that $\ma$ is pRIP, and that the planted matrix $\mg$ in Definition~\ref{def:prip} is an entrywise random Gaussian matrix.

We next consider a natural family of iterative algorithms for solving the system $\ma \sx = b$.  For simplicity, assume that our current iterate is $x_t = 0$.  Inspired by standard gradient methods for solving regression, we consider ``first-order'' algorithms of the following form:
\begin{equation}\label{eq:algfamily}
\begin{aligned}
    y_t &\gets x_t + \ma^\top u \\
    x_{t+1} &\gets \mproj(y_t)
\end{aligned}
\end{equation}
where $u \in \R^n$ are coefficients to be computed (for example, a natural attempt could be to set $u$ to be a multiple of $\ma x_t - b$, resulting in the step being a scaled gradient of $\norm{\ma x - b}_2^2$ at $x_t$) and $\mproj$ denotes projection onto the (convex) $\ell_1$ ball of radius $\sqrt{s}$. Our goal will be to make constant-factor progress in terms of $\norm{x_t - \sx}_2$ in each application of \eqref{eq:algfamily}, to yield a $\tO(1)$ iteration method.  Note that $x_{t + 1}$ must at minimum make constant factor progress in the direction $\sx - x_t$ (in terms of decreasing the projection onto this direction) if we hope to make constant factor progress in overall distance to $\sx$.  In other words, we must have 
\[
\langle x_{t+1}  - x_t,  \sx - x_t \rangle =\Omega(\norm{\sx - x_t}_2^2)  = \Omega(1)\,.
\]
First, observe that by the definition of our step and shifting so that $x_t = 0$, the point $y_t$ satisfies
\[
\langle y_t - x_t, \sx - x_t \rangle = \langle \ma^\top u , \sx - x_t \rangle = u^\top (b - \ma x_t) = u^\top b,
\]
so to obtain a corresponding progress lower bound for the move to $y_t$, we require
\begin{equation}\label{eq:progress-condition}
u^\top b = \Omega(1).
\end{equation}
Of course, this condition alone is not enough, for two reasons: the step also moves in directions orthogonal to $\sx - x_t$, and we have not accounted for the projection step $\mproj$. If all of the rows of $\ma$ were random, then standard Gaussian concentration implies that we expect $\norm{b}_2 = \sqrt{\frac{n}{d}}$, and thus to satisfy \eqref{eq:progress-condition} we need $\norm{w}_2 \geq \sqrt{\frac{d}{n}}$.  This also implies 
\[
\norm{y_t - x_t}_2 = \norm{\ma^\top w}_2 \approx \sqrt{\frac{d}{n}},
\]
since for Gaussian $\ma$, we expect that its rows are roughly orthogonal.  Moreover, since $\sqrt{\frac{d}{n}} \gg 1$ in the typical underconstrained setting, almost all of the step from $x_t$ to $y_t$ is actually orthogonal to the desired ``progress direction''  parallel to $\sx - x_t$.  This appears to be a serious problem, because in order to argue that our algorithm makes progress, we need to argue that the $\ell_1$ projection step $x_{t+1} = \mproj(y_t)$ ``undoes" this huge $\ell_2$ movement orthogonal to the progress direction (see Figure~1).

Our key geometric insight is that the $\ell_1$ ball is very thin in most directions, but is thick in directions that correspond to ``numerically sparse'' vectors, namely vectors with bounded $\ell_1$ to $\ell_2$ ratios.  Crucially, the movement of our step in the progress direction parallel to $\sx - x_t$ is numerically sparse because $\sx - x_t$ is itself $O(s)$-numerically sparse by assumption. However, for Gaussian $\ma$, the motion in the subspace orthogonal to the progress direction ends up being essentially random (and thus is both not sparse, and is $\ell_\infty$-bounded). Formally, we leverage this decomposition to show that the $\ell_1$ projection keeps most of the forward movement in the progress direction $\sx - x_t$ but effectively filters out much of the orthogonal motion, as demonstrated by the figure below.

\begin{figure}[t]\label{fig:projecting}
\centering
\begin{tikzpicture}[ scale = 0.4,point/.style={circle,fill=black, minimum size=0.02cm},extended dashed/.style={dashed, shorten >=-#1,shorten <=-#1},
 extended dashed/.default=2cm]
    \node[point, label = below:{$x_t$}] (xt) at (0,0)   {};
    \node[point,label = below:{$\sx$}] (x*) at  (10,0){} ;
    \node[point,label = below:{$x_{t+1}$}] (xt+1) at  (4, 1.5){};  
    \node[point,label = above:{$y_t$}] (yt) at  (5 , 15 ){} ;
    \draw[extended dashed] (xt+1) --node[yshift = 0.4cm , xshift = 0.4cm] {}  ++ (x*) ;
    \draw (xt) -- node[xshift = -0.4cm,yshift =  0.4cm] {$\sqrt{\frac{d}{n}}$ } ++(yt);
    \draw (xt+1) -- (yt);
    \draw (xt) -- node[yshift =  - 0.4cm] {$1$ } ++(x*);
  \end{tikzpicture}
     \caption{The effect of $\ell_1$ projection on iterate progress. The dashed line represents a facet of the $\ell_1$-ball around $x_t$ of radius $\norm{x_t - \sx}_1$.}
  \end{figure}
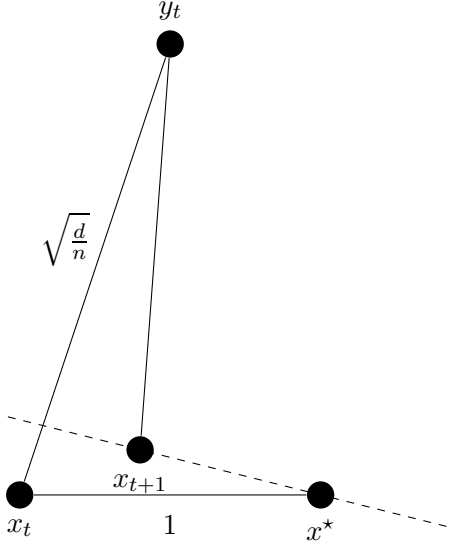

This geometric intuition is the basis for the deterministic conditions we require of the steps of our iterative method, to guarantee that it is making progress. More precisely, the main condition we require of our step in each iteration, parameterized by the coefficients $u$ used to induce \eqref{eq:algfamily}, is that $y_t - x_t$ has a ``short-flat'' decomposition into two vectors $p + e$ where
\[
\norm{p}_2 = O(1) \text{ is ``short'' and } \norm{e}_\infty = O\Par{\frac 1 {\sqrt s}} \text{ is ``flat''.}
\]
The above bounds are rescaled appropriately in our actual method. We state these requirements formally (in the clean observation case, for example) in Definition~\ref{def:ostep}, where we define a ``step oracle'' which is guaranteed to make constant-factor progress towards $\sx$. By combining the above short-flat decomposition requirement with a progress requirement such as \eqref{eq:progress-condition}, we can show that as long as we can repeatedly implement a satisfactory step, our method is guaranteed to converge rapidly.

This framework for sparse recovery effectively reduces the learning problem to the problem of implementing an appropriate step oracle.  Note that a valid step always exists by only looking at the Gaussian rows.  To complete our algorithm, we give an implementation of this step oracle (i.e.\ actually solving for a valid step) which runs in nearly-linear time even when the data is augmented by a semi-random adversary (that is, our measurement matrix is pRIP or wRIP rather than RIP), and demonstrate that our framework readily extends to the noisy observation setting. Our step oracle is motivated by stochastic gradient methods. In particular, we track potentials corresponding to the progress made by our iterative method and the short-flat decomposition, and show that uniformly sampling rows of a pRIP (or wRIP) matrix $\ma$ and taking steps in the direction of these rows which maximally improve our potentials rapidly implements a step oracle.

\subsection{Related work}

\paragraph{Sparse recovery.} Sparse recovery, and variants thereof, are fundamental statistical and algorithmic problems which have been studied in many of settings, including signal processing~\cite{levy1981reconstruction,santosa1986linear,donoho1989uncertainty,blumensath2009iterative,baraniuk2010model}, and compressed sensing~\cite{candes2005signal,candes2006near,candes2006stable,donoho2006compressed,rudelson2006sparse}.
A full review of the extensive literature on sparse recovery is out of the scope of the present paper; we refer the reader to e.g.~\cite{eldar2012compressed,davenport2012introduction,kutyniok2013theory,ludwig2018algorithms} for more extensive surveys.

Within the literature on sparse recovery, arguably the closest line of work to ours is the line of work which attempts to design efficient algorithms which work when the restricted condition number of the sensing or measurement matrix is large.
Indeed, it is known that many nonconvex methods fail when the restricted condition number of the sensing matrix is far from $1$, which is often the case in applications~\cite{jain2014iterative}.
To address this, several works~\cite{jain2014iterative,ludwig2018algorithms} have designed novel non-convex methods which still converge, when the restricted condition number of the matrix is much larger than $1$.
However, these methods still require that the restricted condition number is constant or bounded, whereas in our setting, the restricted condition number could be arbitrarily large due to the generality of the semi-random adversary assumption.

Another related line of work considers the setting where, instead of having a sensing matrix with rows which are drawn from an isotropic Gaussian, have rows drawn from $\Nor(0, \msig)$, for some potentially ill-conditioned $\msig$~\cite{bickel2009simultaneous,raskutti2010restricted,van2013lasso,jain2014iterative,koltchinskii2014l1,dalalyan2017prediction,zhang2017optimal,bellec2018noise,kelner2021power}.
This setting is related to our semi-random adversarial model, in that the information-theoretic content of the problem does not change, but obtaining efficient algorithms which match the optimal statistical rates is very challenging.
However, there does not appear to be any further concrete connection between this ``ill-conditioned covariance'' setting and the semi-random model we consider in this paper.
Indeed, the ill-conditioned setting appears to be qualitatively much more difficult for algorithms: in particular,~\cite{kelner2021power} shows evidence that there are in fact no efficient algorithms that achieve the optimal statistical rates, without additional assumptions on $\msig$.
In contrast in the semi-random setting, polynomial-time convex programming approaches, while having potentially undesirable superlinear runtimes, still obtain optimal statistical guarantees.

Finally as discussed earlier in the introduction, there is a large body of work on efficient algorithms for sparse recovery in an RIP matrix (or a matrix satisfying weaker or stronger analogous properties). These works e.g.\ \cite{candes2005signal,candes2006near,candes2006stable,mallat1993matching,pati1993orthogonal,needell2010signal,needell2009cosamp,blumensath2009iterative,blumensath2010normalized,maleki2010optimally,foucart2011hard,figueiredo2003algorithm,daubechies2004iterative,combettes2005signal,beck2009fast,becker2011nesta,negahban2012unified,agarwal2012fast} are typically based on convex programming or different iterative first-order procedures.

\paragraph{Semi-random models.} Semi-random models were originally introduced in a sequence of innovative papers~\cite{blum1995coloring,feige2001heuristics} in the context of graph coloring.
In theoretical computer science, semi-random models have been explored in many settings, for instance, for various graph-structured~\cite{feige2000finding,feige2001heuristics,chen2012clustering,makarychev2012approximation,makarychev2014constant} and constraint satisfaction problems~\cite{kolla2011play}. 
More recently, they have also been studied for learning tasks such as clustering problems and community detection~\cite{elsner2009bounding,mathieu2010correlation,makarychev2013sorting,chen2014clustering,globerson2014tight,makarychev2015correlation,makarychev2016learning,moitra2016robust}, matrix completion~\cite{cheng2018non}, and linear regression~\cite{li2020wellconditioned}.
We refer the reader to~\cite{roughgarden2021beyond} for a more thorough overview of this vast literature. Finally, we remark that our investigation of the semi-random sparse recovery problem is heavily motivated by two recent works \cite{cheng2018non, li2020wellconditioned} which studied the robustness of \emph{fast iterative methods} to semi-random modeling assumptions.

We also note that the well-studied \emph{Massart noise} model in PAC learning~\cite{massart2006risk} can be thought of as a semi-random variant of the random classification noise model.
However, this setting appears to be quite different from ours: in particular, it was not until quite recently that polynomial-time algorithms were even known to be achievable for a number of fundamental learning problems under Massart noise~\cite{diakonikolas2019distribution,diakonikolas2020learning,chen2020classification,diakonikolas2020hardness,diakonikolas2021threshold,diakonikolas2021forster,diakonikolas2021relu,diakonikolas2021boosting,zhang2021improved}.     	 	%

\section{Preliminaries}
\label{sec:prelims}

\paragraph{General notation.} We let $[n] \defeq \{i \in \N, 1 \le i \le n\}$. The $\ell_p$ norm of a vector is denoted $\norm{\cdot}_p$, and the sparsity (number of nonzero entries) of a vector is denoted $\norm{\cdot}_0$.  For a vector $v \in \R^d$ and $k \in [d]$, we let $\norm{v}_{2, (k)}$ be the $\ell_2$ norm of the largest $k$ entries of $v$ in absolute value (with other elements zeroed out). The all-zeroes vector of dimension $n$ is denoted $0_n$.  The nonnegative probability simplex in dimension $n$ (i.e.\ $\norm{p}_1 = 1$, $p \in \R^n_{\ge 0}$) is denoted $\Delta^n$.

For mean $\mu \in \R^d$ and positive semidefinite covariance $\msig \in \R^{d \times d}$, $\Nor(\mu, \msig)$ denotes the corresponding multivariate Gaussian. $i \sim_{\textup{unif.}} S$ denotes a uniform random sample from set $S$. For $N \in \N$ and $p \in \Delta^n$ we we use $\textup{Multinom}(N, p)$ to denote the probability distribution corresponding to $N$ independent draws from $[n]$ as specified by $p$.

\paragraph{Sparsity.} We say $v$ is \emph{$s$-sparse} if $\norm{v}_0 \le s$. We define the \emph{numerical sparsity} of a vector by
$\NS(v) \defeq \norm{v}_1^2 / \norm{v}_2^2$. Note that from the Cauchy-Schwarz inequality, if $\norm{v}_0 \le s$, then $\NS(v) \le s$.

\paragraph{Matrices.} Matrices are in boldface throughout. The zero and identity matrix of appropriate dimension from context are $\mzero$ and $\id$. For a matrix $\ma \in \R^{n \times d}$, we let its rows be $\ai$, $i \in [n]$ and its columns be $\aj$, $j \in [d]$. The set of $d \times d$ symmetric matrices is $\Sym^d$, and its positive definite and positive semidefinite restrictions are $\PD^d$ and $\PSD^d$. We use the Loewner partial order $\preceq$ on $\Sym^d$. The largest entry of a matrix $\ma \in \R^{n \times d}$ is denoted $\norm{\ma}_{\max} \defeq \max_{i \in [n], j \in [d]} |\ma_{ij}|$. When a matrix $\ma \in \R^{n \times d}$ is clear from context, we refer to its rows as $\{a_i\}_{i \in [n]}$.

\paragraph{Short-flat decompositions.} Throughout we frequently use the notion of ``short-flat decompositions.'' We say $v \in \R^d$ has a \emph{$(C_2, C_\infty)$ short-flat decomposition} if $v = p + e$ for some $e \in \R^d$ with $\norm{e}_2 \leq C_2$ and $p \in \R^d$ with $\norm{p}_\infty \leq C_\infty$. Further, we use $\trunc(v, c) \in \R^d$ for $c \in \R_{\ge 0}$ to denote the vector which coordinatewise $[\trunc(v, c) ]_i = \text{sgn}(v_i)\max(|v_i| - c, 0)$ (i.e.\ the result of adding or subtracting at most  $c$ from each coordinate to decrease the coordinate's magnitude). Note that $v \in \R^d$ has a $(C_2, C_\infty)$ short-flat decomposition if and only if  $\norm{\trunc(v, C_\infty)}_2 \le C_2$ (in which case $p = \trunc(v, C_\infty)$ and $e = v - p$ is such a decomposition).

\paragraph{Restricted isometry property.} We say that matrix $\ma \in \R^{n \times d}$ satisfies the \emph{$(s, c)$-restricted isometry property (RIP)} or (more concisely)  \emph{$\ma$ is $(s, c)$-RIP},  if for all $s$-sparse vectors $v \in \R^d$, 
\[\frac 1 c \norm{v}_2^2 \le \norm{\ma v}_2^2 \le c \norm{v}_2^2.\] 	%

\section{Exact recovery}
\label{sec:exact}

In this section, we give an algorithm for solving the underconstrained linear system $\ma \sx = b$ given the measurement matrix $\ma \in \R^{n \times d}$ (for $n \le d$) and responses $b \in \R^n$ (i.e.\ noiseless or ``exact'' regression), and $\sx$ is $s$-sparse. Our algorithm succeeds when $\ma$ is weighted RIP (wRIP), i.e.\ it satisfies Definition~\ref{def:wrip}, a weighted generalization of Definition~\ref{def:prip}.

\begin{definition}[wRIP matrix]\label{def:wrip} 
	Let $\wsinf \in [0, 1]$. We say $\ma \in \R^{n \times d}$ is $(\rho, \wsinf)$-wRIP if $\norm{\ma}_{\max} \le \rho$, and there exists a weight vector $w^\star \in \Delta^n$ satisfying $\norm{w^\star}_\infty \le \wsinf$, such that $\diag{w^\star}^{\half} \ma$ is $(\Theta(s), \Theta(1))$-RIP for appropriate constants. When $\rho = \tO(1)$ for brevity we say $\ma$ is $\wsinf$-wRIP.
\end{definition}

As discussed after Definition~\ref{def:prip}, a wRIP matrix can be thought of as arising from a ``semi-random model'' because it strictly generalizes our previously-defined pRIP matrix notion in Definition~\ref{def:prip} with $\wsinf = \frac 1 m$, by setting $w^\star$ to be $\frac 1 m$ times the zero-one indicator vector of rows of $\mg$. The main result of this section is the following theorem regarding sparse recovery with wRIP matrices.

\begin{restatable}{theorem}{restateexact}\label{thm:exact}
	Let $\delta \in (0, 1)$, $r > 0$, and suppose $R_0 \ge \norm{\sx}_2$ for $s$-sparse $\sx \in \R^d$. Then with probability at least $1 - \delta$, Algorithm~\ref{alg:hrs} using Algorithm~\ref{alg:ostep} as a step oracle takes as input a $(\rho, \wsinf)$-wRIP matrix $\ma \in \R^{n \times d}$ and $b = \ma \sx$, and computes $\hat{x}$ satisfying
	$\norm{\hat{x}- \sx}_2 \le r$ in time
	\[O\Par{\Par{nd \log^3(nd\rho)\log \Par{\frac {1} {\delta} \cdot \log \frac{R_0}{r}} \log \Par{\frac{R_0}{r}}} \cdot \Par{\wsinf s\rho^2 \log d}}.\]
\end{restatable}

Under the wRIP assumption, Theorem~\ref{thm:exact} provides a natural interpolation between the fully random and semi-random generative models. To build intuition, if a pRIP matrix contains a planted RIP matrix with $\tO(s)$ rows (the information-theoretically minimum size), then by setting $w^\star_\infty \approx \frac 1 {\tO(s)}$, we obtain a near-linear runtime of $\tO(nd)$. However, in the fully random regime where $w^\star_\infty \approx \frac 1 n$ (i.e.\ all of $\ma$ is RIP), the runtime improves $\tO(sd)$ which is sublinear for $n \gg s$. 

The roadmap of our algorithm and its analysis are as follows.

\begin{enumerate}
	\item In Section~\ref{ssec:onephase}, we give an algorithm (Algorithm~\ref{alg:hrs}) which iteratively halves an upper bound on the radius to $\sx$, assuming that either an appropriate step oracle (see Definition~\ref{def:ostep}) based on short-flat decompositions can be implemented for each iteration, or we can certify that the input radius bound is now too loose. This algorithm is analyzed in Lemma~\ref{lem:combinealg}.
	\item We state in Assumption~\ref{assume:det} a set of conditions on a matrix-vector pair $(\ma, \Delta)$ centered around the notion of short-flat decompositions, which suffice to provide a sufficient step oracle implementation with high probability in nearly-linear time. In Section~\ref{ssec:oracle} we analyze this implementation (Algorithm~\ref{alg:ostep}) in the proof of Lemma~\ref{lem:oracleunderdet} assuming the inputs satisfy Assumption~\ref{assume:det}.
	\item In Section~\ref{ssec:detassume}, we show Assumption~\ref{assume:det}, with appropriate parameters, follows from $\ma$ being wRIP. This is a byproduct of a general equivalence we demonstrate between RIP, restricted conditioning measures used in prior work \cite{AgarwalNW10}, and short-flat decompositions.
\end{enumerate}

\subsection{Radius contraction using step oracles}\label{ssec:onephase}

In this section, we provide and analyze the main loop of our overall algorithm for proving Theorem~\ref{thm:exact}. This procedure, $\HRS$, takes as input an $s$-sparse vector $\xin$ and a radius bound $R \ge \norm{\xin - \sx}_2$ and returns an $s$-sparse vector $\xout$ with the guarantee $\norm{\xout - \sx}_2 \le \half R$. As a subroutine, it requires access to a ``step oracle'' $\ostep$, which we implement in Section~\ref{ssec:oracle} under certain assumptions on the matrix $\ma$. 

\begin{definition}[Step oracle]\label{def:ostep}
We say that $\ostep$ is a $(\Cp, C_2, \delta)$-step oracle for $\Delta \in \R^n$ and $\ma \in \R^{n \times d}$, if the following holds. Whenever there is $v \in \R^d$ with $\frac 1 4 \le \norm{v}_2 \le 1$ and $\norm{v}_1 \le 2\sqrt{2s}$ such that $\Delta = \ma v$, with probability $\ge 1 - \delta$, $\ostep$ returns $w \in \R^n_{\ge 0}$ such that the following two conditions hold. First,
\begin{equation}\label{eq:prog}\sum_{i \in [n]} w_i \Delta_i^2 \ge \Cp.\end{equation}
Second, there exists a $(C_2, \frac{\Cp}{6\sqrt s})$ short-flat decomposition of $\ma^\top \diag{w} \Delta$:
\begin{equation}\label{eq:shortflat}\norm{\trunc\Par{\ma^\top \diag{w} \Delta, \frac{\Cp}{6\sqrt s}}}_2 \le C_2.\end{equation}
\end{definition}

Intuitively, \eqref{eq:shortflat} guarantees that we can write $\gamma = p + e$ where $p$ denotes a ``progress'' term which we require to be sufficiently short in the $\ell_2$ norm, and $e$ denotes an ``error'' term which we require to be small in $\ell_\infty$. We prove that under certain assumptions on the input $\ma$ (stated in Assumption~\ref{assume:det} below), we can always implement a step oracle with appropriate parameters.

\begin{restatable}{assumption}{restateassumedet}\label{assume:det}
	The matrix $\ma \in \R^{n \times d}$ satisfies the following. There is a weight vector $w^\star \in \Delta^n$ satisfying $\norm{w^\star}_\infty \le \wsinf$, a constant $L$, $\rho \ge 1$, and a constant $K$ (which may depend on $L$) such that for all $v \in \R^d$ with $\frac 1 4 \le \norm{v}_2 \le 1$ and $\norm{v}_1 \le 2 \sqrt{2s}$ we have, defining $\Delta = \ma v$:
	\begin{enumerate}
		\item $\ma$ is entrywise bounded by $\pm \rho$, i.e.\ $\norm{\ma}_{\max} \le \rho$.
		\item
		\begin{equation}\label{eq:rsc}\frac 1 L \le \sum_{i \in [n]} w^\star_i \Delta_i^2 \le L.\end{equation}
		\item For $\mw^\star \defeq \diag{w^\star}$, there is a $(L, \frac{1}{K\sqrt s})$ short-flat decomposition of $\ma^\top \mw^\star \Delta = \sum_{i \in [n]} w^\star_i \Delta_i a_i$: %
		\begin{equation}\label{eq:shortflatv}\norm{\trunc\Par{\ma^\top \mw^\star \Delta, \frac{1}{K \sqrt s}}}_2 \le L.\end{equation}
	\end{enumerate}	
\end{restatable}

Our Assumption~\ref{assume:det} may also be stated in a scale-invariant way (i.e.\ with \eqref{eq:rsc}, \eqref{eq:shortflatv} scaling with $\norm{v}_2$), but it is convenient in our analysis to impose a norm bound on $v$. Roughly, the second property in Assumption~\ref{assume:det} is (up to constant factors) equivalent to the ``restricted strong convexity'' and ``restricted smoothness'' assumptions of \cite{AgarwalNW10}, which were previously shown for specific measurement matrix constructions such as random Gaussian matrices. The use of the third property in Assumption~\ref{assume:det} (the existence of short-flat decompositions for numerically sparse vectors) in designing an efficient algorithm is a key contribution of our work. Interestingly, we show in Section~\ref{ssec:detassume} that these assumptions are up to constant factors equivalent to RIP.

More specifically, we show that when $\ma$ is wRIP, we can implement a step oracle for $\Delta = \ma v$ where $v = \frac 1 R (x - \sx)$ for some iterate $x$ of Algorithm~\ref{alg:hrs}, which either makes enough progress to advance the algorithm or certifies that $v$ is sufficiently short, by using numerical sparsity properties of $v$. We break this proof into two parts. In Lemma~\ref{lem:oracleunderdet}, we show that Assumption~\ref{assume:det} suffices to implement an appropriate step oracle; this is proven in Section~\ref{ssec:oracle}. In Lemma~\ref{lem:detundergen}, we then demonstrate the wRIP assumption with appropriate parameters implies our measurement matrix satisfies Assumption~\ref{assume:det}, which we prove by way of a more general equivalence in Section~\ref{ssec:detassume}.

\begin{restatable}{lemma}{restatesteporacle}\label{lem:oracleunderdet}
Suppose $\ma$ satisfies Assumption~\ref{assume:det}. Algorithm~\ref{alg:ostep} is a $(\Cp, C_2, \delta)$ step oracle $\OStep$ for $(\Delta, \ma)$ with $\Cp = \Omega(1)$, $C_2 = O(1)$ running in time
\[O\Par{\Par{ nd\log^3 (nd\rho)\log \frac 1 \delta} \cdot \Par{\wsinf  s\rho^2\log d}}.\]
\end{restatable}

\begin{restatable}{lemma}{restatedetundergen}\label{lem:detundergen}
	Suppose $\ma \in \R^{n \times d}$ is $(\rho, \wsinf)$-wRIP with a suitable choice of constants in the RIP parameters in Definition~\ref{def:wrip}. Then, $\ma$ also satisfies Assumption~\ref{assume:det}.
\end{restatable}

We now give our main algorithm $\HRS$, assuming access to the step oracle $\ostep$ from Section~\ref{ssec:oracle} with appropriate parameters, and that $\ma$ obeys Assumption~\ref{assume:det}.

\begin{algorithm2e}[ht!]
	\caption{$\HRS(x_{\text{in}}, R, \ostep, \delta, \ma, b)$}
	\label{alg:hrs}
	\codeInput $s$-sparse $\xin \in \R^d$, $R \ge \norm{\xin - \sx}_2$ for $s$-sparse $\sx \in \R^d$, $(\Cp, C_2, \delta)$-step oracle $\ostep$ for all $(\Delta, \ma)$ with $\Delta \in \R^n$, $\delta \in (0, 1)$, $\ma \in \R^{n \times d}$, $b = \ma \sx\in \R^n$ \;
	\codeOutput $s$-sparse vector $x_{\text{out}}$ that satisfies 
	$\norm{x_{\text{out}} - \sx}_2 \le \half R$ with probability $\ge 1 - T\delta$ \;
	\codeLineSace
	
	Set	$x_0 \gets x_{\text{in}}$, $\xset \gets \{x \in \R^d \mid \norm{x - \xin}_1 \le \sqrt{2s} R\}$
	$T \gets \left\lceil \frac{6C_2^2}{\Cp^2} \right\rceil$, $\eta \gets \frac{\Cp}{2C_2^2}$ \;
	\For{$0 \le t \le T - 1$}{
		$w_t \gets \ostep(\Delta_t,  \ma)$ for $\Delta_t \gets\frac{1}{R}(\ma x_t - b)$, $\gamma_t \gets \ma^\top \diag{w_t} \Delta_t = \sum_{i \in [n]} [w_t]_i [\Delta_t]_i a_i$  \;
		\If{$\sum_{i\in[n]} [w_t]_i [\Delta_t]_i^2 < \Cp$ $\textup{or}$ $\norm{\trunc(\gamma_t,\frac{\Cp}{6\sqrt{s}})}_2 > C_2$ \label{line:half_if_check}}{
			\codeReturn $\xout \gets x_t$ truncated to its $s$ largest coordinates
		}
		\lElse{
			$x_{t + 1} \gets \argmin_{x \in \xset} \norm{x - x_t - \eta R \gamma_t}_2$
		}
	}
	\codeReturn $\xout \gets x_t$ truncated to its $s$ largest coordinates
\end{algorithm2e}

\begin{lemma}\label{lem:combinealg}
Assume $\ma$ satisfies Assumption~\ref{assume:det}. With probability at least $1 - T\delta$, Algorithm~\ref{alg:hrs} succeeds (i.e.\ $\norm{x_{\text{out}} - \sx}_2 \le \half R$).
\end{lemma}
\begin{proof}
Throughout this proof, condition on the event that all step oracles succeed (which provides the failure probability via a union bound). We first observe that $\sx \in \xset$ because of Cauchy-Schwarz, the $2s$-sparsity of $\xin - \sx$, and the assumption $\norm{\xin - \sx}_2 \le R$.

Next, we show that in every iteration $t$ of Algorithm~\ref{alg:hrs},
\begin{equation}\label{eq:singlestep}\norm{x_{t + 1} - \sx}_2^2 \le \Par{1 - \frac{C_2^2}{2\Cp^2}}\norm{x_t - \sx}_2^2.\end{equation}
As $\sx \in \xset$, the optimality conditions of $x_{t + 1}$ as minimizing $\norm{x - (x_t - \eta R \gamma_t)}_2^2$ over $\xset$ imply
\begin{equation}\label{eq:optconds}
	\begin{gathered}2\inprod{x_{t + 1} - x_t + \eta R \gamma_t}{x_{t + 1} - \sx} \le 0 \\
		\implies \norm{x_t - \sx}_2^2 - \norm{x_{t + 1} - \sx}_2^2 \ge 2\eta R \inprod{\gamma_t}{x_{t + 1} - \sx} + \norm{x_t - x_{t + 1}}_2^2. \end{gathered}
\end{equation}
Hence, it suffices to lower bound the right-hand side of the above expression. Let $\gamma_t = p_t + e_t$ denote the  $(C_2, \frac{\Cp}{6\sqrt s})$ short-flat decomposition of $\gamma_t$ which exists by Definition~\ref{def:ostep} assuming the step oracle succeeded. We begin by observing
\begin{equation}\label{eq:otherterms}
	\begin{aligned}
		2\eta R \inprod{\gamma_t}{x_{t + 1} - x_t} + \norm{x_t - x_{t + 1}}_2^2 &= 2\eta R \inprod{e_t}{x_{t + 1} - x_t} + 2\eta R \inprod{p_t}{x_{t + 1} - x_t} + \norm{x_t - x_{t + 1}}_2^2 \\
		&\ge -2\eta R\norm{e_t}_\infty \norm{x_{t + 1} - x_t}_1 - \eta^2 R^2 \norm{p_t}_2^2 \\
		&\ge -\eta R^2 \Cp - \eta^2 R^2 C_2^2.
	\end{aligned}
\end{equation}
The first inequality followed from H\"older on the first term and Cauchy-Schwarz on the latter two terms in the preceding line. The second followed from the $\ell_1$ radius of $\xset$, and the bounds on $e_t$ and $p_t$ from \eqref{eq:shortflat}. Next, from Definition~\ref{def:ostep}, for $\Delta = \Delta_t = \frac{1}{R}(\ma x_t - b)$ and $v = \frac{1}{R}(x_t - \sx)$,
\begin{equation}\label{eq:progressterm}
	2\eta R \inprod{\gamma_t}{x_t - \sx} = 2 \eta R \sum_{i \in [n]} w_i \Delta_i \inprod{a_i}{v} = 2\eta R^2 \sum_{i \in [n]} w_i \Delta_i^2 \ge 2\eta R^2 \Cp.
\end{equation}
Finally, \eqref{eq:singlestep} follows from combining \eqref{eq:optconds}, \eqref{eq:otherterms}, and \eqref{eq:progressterm}, with our choice of $\eta$, and the fact that inducting on this lemma implies the $\ell_2$ distance to $\sx$ of the iterates is monotone decreasing.

Next, we claim that regardless of whether Algorithm~\ref{alg:hrs} terminates on Line 7 or Line 11, we have $\norm{x_t - \sx}_2 \le \frac 1 4 R$. Note that the vector $v = \frac{1}{R} (x_t - \sx)$ satisfies $\ma v = \Delta \defeq \frac{1}{R}(\ma x_t - b)$. By assumption the condition $\norm{v}_1 \le 2 \sqrt{2s}$ is met (since $x_t, x^\star \in \xset$), and upon iterating \eqref{eq:singlestep} on our radius bound assumption, this implies that the condition $\norm{v}_2 \le 1$ is met. Hence, if the algorithm terminated on Line 7, we must have $\norm{v}_2 \le \frac 1 4 R \implies \norm{x_t - \sx}_2 \le \frac 1 4 R$, as otherwise the termination condition would have been false. On the other hand, by \eqref{eq:singlestep}, after $T$ steps we have
\[\norm{x_{T} - \sx}_2^2 \le \exp\Par{-\frac{TC_2^2}{2\Cp^2}} \norm{x_0 - \sx}_2^2 \le \frac{1}{16} R^2.\]
We conclude that at termination, $\norm{x_t - \sx}_2 \le \frac 1 4 R$. Now, $s$-sparsity of $\sx$ and the definition of $\xout = \argmin_{\norm{x}_0 \le s} \norm{x - x_t}_2$ imply the desired 
\begin{equation}\label{eq:roundingdouble}
	\norm{\xout - \sx}_2 \le \norm{\xout - x_t}_2 + \norm{\sx - x_t}_2 \le 2\norm{\sx - x_t}_2 \le \half R.
\end{equation}
\end{proof}

\subsection{Designing a step oracle}\label{ssec:oracle}

In this section, we design a step oracle $\ostep(\Delta, \ma)$ (see Definition~\ref{def:ostep}) under Assumption~\ref{assume:det} on the input matrix $\ma \in \R^{n \times d}$. Our step oracle iteratively builds a weight vector $\bw \in \R^n_{\ge 0}$. It will be convenient to define
\begin{equation}\label{eq:gammadef}
\gamma_{\bw} \defeq \sum_{i \in [n]} \bw_i \Delta_i a_i.
\end{equation}
Note that a valid step oracle always exists (although it is unclear how to implement the following solution): namely, setting $\bw = w^\star$ satisfies the oracle assumptions by the second and third conditions in Assumption~\ref{assume:det}. 
In order to ensure Algorithm~\ref{alg:onstep} is indeed a step oracle, we track two potentials for some $\mu$, $C$ we will define in Algorithm~\ref{alg:ostep}:
\begin{equation}\label{eq:abdef}
\begin{gathered}
\Apot(\bw) \defeq \sum_{i \in [n]} \bw_i \Delta_i^2 \text{ and } \; \Bpot(\bw) \defeq \Par{\min_{\norm{p}_2 \le L\norm{\bw}_1}\sqmax_\mu\Par{\gamma_{\bw} - p}} + \frac{\norm{\bw}_1}{4CLs},\\
\text{where } \sqmax_\mu(x) \defeq \mu^2 \log\Par{\sum_{j \in [d]} \exp\Par{\frac{x_j^2}{\mu^2}} }.
\end{gathered}
\end{equation}
Intuitively, $\Apot(\bw)$ corresponds to progress on \eqref{eq:prog}, and $\Bpot(\bw)$ is intended to track the bounds \eqref{eq:shortflat}. We note the following fact about the $\sqmax$ function which follows from direct calculation.

\begin{fact}\label{fact:smax}
For all $x \in \R^d$, $\norm{x}_\infty^2 \le \sqmax_\mu(x)$, and $\sqmax_\mu(x) \ge \mu^2 \log(d)$.
\end{fact}

Also it will be important to note that $\Bpot(\bw)$ can be computed to high precision efficiently. We state this claim in the following and defer a full proof to Appendix~\ref{app:deferred}; we give a subroutine which performs a binary search on a Lagrange multiplier on the $\ell_2$ constraint on $p$, and then solves for each optimal $p_j$ using another binary search based on the Lagrange multiplier value.

\begin{restatable}{lemma}{restatesqmax}\label{lem:optimize-softmax}
Let $\delta > 0$ and $\theta \ge 0$.  For any vector $\gamma \in \R^d$, we can solve the optimization problem 
\[
\min_{\norm{p}_2 \leq \theta} \sqmax_\mu( \gamma - p)
\]
to additive accuracy $\delta$ in time 
\[O\left(d \log^2\left( \frac{\norm{\gamma}_2^2}{\mu \sqrt{\delta}}\right) 
\right).\]
\end{restatable}

We state the full implementation of our step oracle as Algorithm~\ref{alg:ostep} below.

\begin{algorithm2e}[ht!]
	\caption{$\OStep(\Delta, \ma, \delta)$}
	\label{alg:ostep}
		\codeInput $\Delta \in \R^n, \ma \in \R^{n \times d}$ satisfying Assumption~\ref{assume:det}, $\delta \in (0, 1)$ \;
		\codeOutput $w$ such that if there is $v \in \R^d$ with $\frac 1 4 \le \norm{v}_2 \le 1$ and $\norm{v}_1 \le 2\sqrt{2s}$ such that $\Delta = \ma v$, with probability $\ge 1 - \delta$, \eqref{eq:prog}, \eqref{eq:shortflat} are satisfied with $Cp = 1,\; C_2 = O(1).$
		\codeLineSace
		
		$C \gets 200$, $\mu \gets \frac{1}{\sqrt{Cs \log d} }$, $\eta \gets \frac 1 {K\wsinf s\rho^2\log d}$, $N' \gets \lceil\log_2 \frac 1 \delta\rceil$ \;
		\For{$0 \le  k \le N'$}{
			$w_0 \gets 0_n$, $N \gets \lceil \frac{5Ln}{\eta}\rceil$\;
			\For{$0 \le t \le N$}{
				\lIf{$\Apot(w_t) \ge 1$}{
					\codeReturn $w \gets w_t$ 
				}
				Sample $i \sim_{\textup{unif.}} [n]$ \;
				Compute (using Lemma~\ref{lem:optimize-softmax}) $d_t \in [0, \eta\wsinf]$ maximizing to additive $O(\frac \eta n)$
		\begin{equation}\label{eq:deltadef}\Gamma_t(d) \defeq \Apot(w_t + d e_i) - Cs \Bpot(w_t + d e_i)\end{equation}
				$w_{t + 1} \gets w_t + d_t e_i$ \;
			}
		}
		\codeReturn $w \gets 0_n$ \;
\end{algorithm2e}

Our main helper lemma bounds the expected increase in $\Bpot$ from choosing a row of $\ma$ uniformly at random, and choosing a step size according to $w^\star$. We do not know $w^\star$, but we argue that our algorithm makes at least this much expected progress. Define the decomposition promised by \eqref{eq:shortflatv}:
\[p^\star \defeq \trunc\Par{\ma^\top \mw^\star \Delta, \frac{1}{K\sqrt s}},\; e^\star \defeq  \ma^\top \mw^\star \Delta - p^\star.\]
Furthermore, define for all $i \in [n]$,
\begin{equation}\label{eq:zdef}z^{(i)} \defeq \eta w^\star_i(\Delta_i a_i - \pstar),\end{equation}
where $\pstar$ is given by \eqref{eq:shortflatv}. We use $\{z^{(i)}\}_{i \in [n]}$ as certificates of $\Bpot$'s growth in the following.

\begin{lemma}\label{lem:bgrowth}
Assume that the constant $K$ in Assumption~\ref{assume:det} is sufficiently large, and that $\Delta = \ma v$ where $v$ satisfies the norm conditions in Assumption~\ref{assume:det}. Then for any $\bw \in \R^n_{\ge 0}$ such that $\Bpot(\bw) \leq C^2 \mu^2 \log d$, and $\eta \le \frac{1}{K\wsinf  s\rho^2\log d}$, we have
\[ \E_{i \sim_{\textup{unif.}} [n]} \Brack{\Bpot(\bw + \eta w^\star_i)} \le \Bpot(\bw) + \frac{1}{2CLs}\cdot \frac{\eta}{n}. \]
\end{lemma}
\begin{proof}
We assume for simplicity $L \ge 2\sqrt{2}$ as otherwise we may set $L \gets \max(2\sqrt{2}, L)$ and \eqref{eq:rsc} remains true. Let $p_{\bw}$ be the minimizing argument in the definition of $\Bpot(\bw)$ in \eqref{eq:abdef}. For any $i \in [n]$, it is clear that $p_{\bw} + (\eta w_i^\star) \pstar $ is a valid argument for the optimization problem defining $\Bpot(\bw + \eta w^\star_i)$ by $\norm{\pstar}_2 \le L$, and since $\norm{w}_1$ grows by $\eta w_i^\star$. Next, define
\begin{equation}\label{eq:fdef}
F(x) \defeq \sum_{j \in [d]} \exp\Par{\frac{x_j^2}{\mu^2} }
\end{equation}
such that $\Bpot(\bw) = \mu^2 \log F(x) + \frac{\norm{\bw}_1}{4CLs}$ for $x = \gamma_{\bw} - p_{\bw}$. As discussed earlier, since $\norm{p_{\bw} + (\eta w_i^\star)p^\star}_2 \le \norm{p_{\bw}}_2 + \eta w_i^\star L$, we conclude
\begin{equation}\label{eq:bpotvalid}
\Bpot(\bw + \eta w^\star_i) \le \mu^2 \log F(x + z^{(i)}) + \frac{\norm{\bw + \eta w^\star_i}_1}{4CLs}.
\end{equation} 
We next compute
\begin{equation}\label{eq:avFdiff}
\begin{aligned}
\frac 1 n \sum_{i \in [n]} F(x + z^{(i)}) &= \frac 1 n \sum_{j \in [d]} \exp\Par{\frac{x_j^2}{\mu^2}} \Par{\sum_{i \in [n]} \exp\Par{\frac{2x_j z^{(i)}_j + (z^{(i)}_j)^2}{\mu^2}}}  \\
&\le \frac 1 n F(x) \max_{j \in [d]} \Par{\sum_{i \in [n]} \exp\Par{\frac{2x_j z^{(i)}_j + (z^{(i)}_j)^2}{\mu^2}}}.
\end{aligned}
\end{equation}
We now bound the right-hand side of this expression. For any $i \in [n]$ and $j \in [d]$, recalling \eqref{eq:zdef},
\begin{equation}\label{eq:zinfbound}
\Abs{z^{(i)}_j} \le \eta w_i^\star (|\Delta_i| \norm{a_i}_\infty + \norm{\pstar}_2) \le \eta \wsinf L(\sqrt s \rho^2 + 1).
\end{equation}
The second inequality used our bounds from Assumption~\ref{assume:det}; note that for $\Delta = \ma v$ where $v$ satisfies the norm conditions in Assumption~\ref{assume:det}, $|\Delta_i| \le \rho \norm{v}_1 \le 2\sqrt{2s}\rho$. Hence, if we choose a sufficiently large constant $K$ in Assumption~\ref{assume:det}, we have
\[\frac{1}{\mu} \Abs{z^{(i)}_j} \le \frac{\sqrt{C}}{K\sqrt{s\log d}\rho^2} \cdot \Par{L (\sqrt s \rho^2 + 1)} \le \frac{1}{4C\sqrt{ \log d}} .\]
Also by the assumption that $\Bpot(\bw) \leq C^2 \mu^2 \log d$ we must have that for all $j \in [d]$,
\[
\frac{|x_j|}{\mu} \leq C\sqrt{\log d}.
\]
Now, using $\exp(c) \le 1 + c + c^2$ for $|c| \le 1$, we get
\begin{equation}\label{eq:taylor}
\begin{aligned}
\sum_{i \in [n]}\exp\Par{\frac{2x_j z^{(i)}_j + (z^{(i)}_j)^2}{\mu^2}} &\le \sum_{i \in [n]} \Par{1 + \frac{2x_j z^{(i)}_j}{\mu^2} +    \frac{(z^{(i)}_j)^2}{\mu^2} + \Par{\frac{2x_j z^{(i)}_j + (z^{(i)}_j)^2}{\mu^2}}^2} \\
&\le \sum_{i \in [n]} \Par{1 + \frac{2x_j z^{(i)}_j}{\mu^2} +   10C^2\log d \cdot  \frac{(z^{(i)}_j)^2}{\mu^2}}.
\end{aligned}
\end{equation}
We control the first-order term via the observation that $\sum_{i \in [n]} z^{(i)} = \eta \estar$ which is $\ell_\infty$-bounded from \eqref{eq:shortflatv}, so taking the constant $K$ in Assumption~\ref{assume:det} sufficiently large, we have
\begin{equation}\label{eq:linbound}
\begin{gathered}
\left \lvert \sum_{i \in [n]}  \frac{z^{(i)}_j}{\mu} \right \rvert \le \frac \eta \mu \norm{\estar}_\infty \le \frac{\eta\sqrt{C\log d}}{K} \\
\implies \Abs{\sum_{i \in [n]} \frac{2x_j z_j^{(i)}}{\mu^2}} \le 2C\sqrt{\log d} \cdot \frac{\eta \sqrt{C\log d}}{K} \le \frac{\eta \log d}{8L}.
\end{gathered}
\end{equation}
In the last inequality we assumed $K \ge 16C^{1.5}L$.
We control the second-order term by using $(a + b)^2 \le 2a^2 + 2b^2$, $\norm{\pstar}_\infty \le \norm{\pstar}_2 \le L$, and \eqref{eq:rsc}:
\begin{equation}\label{eq:quadbound}
\begin{aligned}
\sum_{i \in [n]} \Par{z^{(i)}_j}^2 &\le 2\eta^2\wsinf \Par{\sum_{i \in [n]} w_i^\star [\pstar]_j^2  +  \sum_{i \in [n]} w_i^\star \Delta_i^2\rho^2} \le 2\eta^2\wsinf (L\rho^2 + L^2).
\end{aligned}
\end{equation}
Putting together \eqref{eq:taylor}, \eqref{eq:linbound}, and \eqref{eq:quadbound}, with the definition of $\mu$, we conclude for sufficiently large $K$,
\begin{align*}
\frac 1 n \sum_{i \in [n]} \exp\Par{\frac{2x_j z^{(i)}_j + (z^{(i)}_j)^2}{\mu^2}} &\le 1 + \frac{\eta \log d}{8Ln} + \frac{2\eta^2\wsinf(L\rho^2 + L^2)}{n\mu^2} \\
&\le 1 + \frac{\eta \log d}{4Ln} .
\end{align*}
Hence, combining the above with \eqref{eq:avFdiff}, and using $\log(1 + c) \le c$ for all $c$,
\begin{equation}\label{eq:Fbound}
\mu^2 \log \Par{\frac 1 n \sum_{i \in [n]} F(x + z^{(i)})} \le \mu^2 \log F(x) + \frac{\mu^2\eta \log d}{4Ln} = \mu^2 \log F(x) + \frac{1}{Cs} \cdot \frac{\eta }{4Ln}. 
\end{equation}
Finally, we compute via \eqref{eq:bpotvalid} and concavity of $\log$,
\begin{align*}\E_{i \sim_{\textup{unif.}} [n]} [\Bpot(\bw + \eta w_i^\star)] &\le \frac{\mu^2}{n} \sum_{i \in [n]} \log F(x + z^{(i)}) + \frac{ \norm{\bw}_1}{4CLs} + \frac{1}{4CLs}\Par{\frac{1}{n} \sum_{i \in [n]} \eta w_i^\star} \\
&\le \mu^2 \log\Par{\frac 1 n \sum_{i \in [n]} F(x + z^{(i)})} + \frac{ \norm{\bw}_1}{4CLs} + \frac{1}{4CLs} \cdot \frac \eta n \\
&\le \mu^2 \log F(x) + \frac{ \norm{\bw}_1}{4CLs} + \frac{1}{Cs} \cdot \frac {\eta }{2Ln} = B(\bw) + \frac{1}{Cs}\cdot \frac{\eta}{2Ln}.\end{align*}
In the last line, we used the bound \eqref{eq:Fbound}.
\end{proof}

Finally, we can complete the analysis of Algorithm~\ref{alg:ostep}.

\restatesteporacle*
\begin{proof} 
	
It suffices to prove Algorithm~\ref{alg:ostep} meets its output guarantees in this time. Throughout this proof, we consider one run of Lines 5-10 of the algorithm, and prove that it successfully terminates on Line 7 with probability $\ge \half$ assuming $\ma$ satisfies Assumption~\ref{assume:det} and that $\Delta = \ma v$ for $v$ satisfying the norm bounds in Assumption~\ref{assume:det}. This yields the failure probability upon repeating $N'$ times. 

For the first part of this proof, we assume we can exactly compute $\Delta_t$, and carry out the proof accordingly. We discuss issues of approximation tolerance at the end, when bounding the runtime. 

\paragraph{Correctness.} We use the notation $A_t \defeq \Apot(w_t)$, $B_t \defeq \Bpot(w_t)$, and $\Phi_t \defeq A_t - CsB_t$. We first observe that $A_t$ is $1$-Lipschitz, meaning it can only increase by $1$ in any given iteration; this follows from $\eta\wsinf \Delta_i^2 \le \frac{1}{8s\rho^2} \Delta_i^2 \le 1$, since $\Delta_i^2 = \inprod{a_i}{v}^2 \le 8s\rho^2$ by $\ell_\infty$-$\ell_1$ H\"older. 

Suppose some run of Lines 5-13 terminates by returning on Line 8 in iteration $T$, for $0 \le T \le N$. The termination condition implies that $A_T \ge 1 = \Cp$, so to show that the algorithm satisfies Definition~\ref{def:ostep}, it suffices to show existence of a short-flat decomposition in the sense of \eqref{eq:shortflat}. Clearly, $\Phi_t$ is monotone non-decreasing in $t$, since we may always force $\Gamma_t = 0$ by choosing $d_t = 0$. Moreover, $\Phi_0 = -CsB_0 = -Cs\mu^2 \log d = -1$. The above Lipschitz bound implies that $A_T \le 2$, since $A_{T - 1} \le 1$ by the termination condition; hence,
\[A_T - CsB_T = \Phi_T \ge \Phi_0 = -1 \implies B_T \le \frac{A_T + 1}{Cs} \le \frac{3}{Cs} \le C^2\mu^2\log d.\]
Note that the above inequality and nonnegativity of $\sqmax_\mu$ imply that $\frac{\norm{w_T}_1}{4LCs} \le \frac{3}{Cs}$, so $\norm{w_T}_1 \le12L$. For the given value of $C = 200$, and the first inequality in Fact~\ref{fact:smax}, the definition of the first summand in $B$ implies there is a short-flat decomposition meeting \eqref{eq:shortflat} with $C_2 = L\norm{w_T}_1 = O(1)$.

Hence, we have shown that Definition~\ref{def:ostep} is satisfied whenever the algorithm returns on Line 7. We make one additional observation: whenever $\Phi_t \ge 0$, the algorithm will terminate. This follows since on such an iteration,
\[A_t \ge CsB_t \ge CsB_0 = Cs\mu^2\log d = 1,\]
since clearly the function $B$ is minimized by the all-zeroes weight vector, attaining value $\mu^2\log d$.

\paragraph{Success probability.} We next show that with probability at least $\half$, the loop in Lines 5-10 will terminate. Fix an iteration $t$. When sampling $i \in [n]$, the maximum gain in $\Phi_t$ for $d_t \in [0, \eta \wsinf]$ is at least that attained by setting $d_t = \eta w^\star_i$, and hence
\begin{equation}\label{eq:phigrowth}\E[\Phi_{t + 1} - \Phi_t \mid A_t \le 1] \ge \frac{\eta}{Ln} - \frac{\eta}{2Ln} = \frac{\eta}{2Ln}.\end{equation}
Here, we used that the expected gain in $A_t$ by choosing $d_t = \eta w^\star_i$ over a uniformly sampled $i \in [n]$ is lower bounded by $\frac{\eta}{Ln}$ via \eqref{eq:rsc}, and the expected gain in $Cs B_t$ is upper bounded by Lemma~\ref{lem:bgrowth}. 

Let $Z_t$ be the random variable equal to $\Phi_t - \Phi_0$, where we freeze the value of $w_{t'}$ for all $t' \ge t$ if the algorithm ever returns on Line 8 in an iteration $t$. Notice that $Z_t \le 2$ always: whenever $Z_t \ge 1$, we have $\Phi_t \ge 0$ so the algorithm will terminate, and $Z_t$ is $1$-Lipschitz because $A_t$ is. Moreover, whenever we are in an iteration $t$ where $\Pr[A_t \ge 1] \le \half$, applying \eqref{eq:phigrowth} implies 
\begin{align*}
\E[Z_{t + 1} - Z_t] = \E[Z_{t + 1} - Z_t \mid A_t \le 1] \Pr[A_t \le 1] \ge \frac{\eta}{4Ln}.
\end{align*} 
Clearly, $\Pr[A_t \ge 1]$ is a monotone non-decreasing function of $t$, since $A_t$ is monotone. After $N \ge \frac{5Ln}{\eta}$ iterations, if we still have $\Pr[A_t \ge 1] \le \half$, we would obtain a contradiction since recursing the above display yields $\E[Z_N] > 2$. This yields the desired success probability.

\paragraph{Runtime.} The cost of each iteration is dominated by the following computation in Line 9: we wish to find $d \in [0, \eta w^\star_\infty]$ maximizing to additive $O(\frac \eta n)$ the following objective:
\[\Apot(w + d e_i) - Cs\Bpot(w + d e_i).\]
We claim the above function is a concave function of $d$. First, we show $\Bpot$ is convex (and the result will then follow from linearity of $\Apot$). To see this, for two values $w_i$ and $w'_i$, let the corresponding maximizing arguments in the definition of $\Bpot(\bw + w_i)$ and $\Bpot(\bw + w'_i)$ be denoted $p$ and $p'$. Then, $\half(p + p')$ is a valid argument for $\bw + \half (w_i + w'_i)$, and by convexity of $\sqmax_\mu$ and linearity of the $\ell_1$ portion, we have the conclusion. 

Next, note that all $|\Delta_i|$ are bounded by $2\sqrt{2s}\rho$ (proven after \eqref{eq:zinfbound}) and all $a_{ij}$ are bounded by $\rho$ by assumption. It follows that the restriction of $\Apot$ to a coordinate is $8s\rho^2$-Lipschitz. Moreover the linear portion of $\Bpot$ is clearly $\frac 1 {4CLs}$-Lipschitz in any coordinate. Finally we bound the Lipschitz constant of the $\sqmax$ part of $\Bpot$. It suffices to bound Lipschitzness for any fixed $p$ of
\[\sqmax_\mu\Par{\gamma_{\bw} - p + d_i \Delta_i a_i}\]
because performing the minimization over $p$ involved in two $\sqmax(\gamma_{\bw} - p + d \Delta_i a_i)$ and $\sqmax(\gamma_{\bw} - p + d' \Delta_i a_i)$ can only bring the function values closer together. By direct computation the derivative of the above quantity with respect to $d_i$ is
\[\sum_{j \in [d]} \Delta_i a_{ij} \Par{2\Brack{\gamma_{\bw} - p + d_i\Delta_i a_i}_j} q_j \]
for some probability density vector $q \in \Delta^d$. Further we have
\[\Abs{\gamma_{\bw} - p + d_i\Delta_i a_i}_j \le O\Par{\sqrt{s}\rho^2} + O(1) + 2\sqrt{2s}\rho^2\cdot (\eta \wsinf).\]
Here we used our earlier proof that we must only consider values of $\norm{p}_2 = O(1)$ throughout the algorithm (since $\norm{w_t}_1 = O(1)$ throughout) and this also implies no coordinate of $\gamma_{\bw}$ can be larger than $(\max_{i \in [n]} |\Delta_i|) (\max_{i \in [n], j \in [d]} |a_{ij}|) \norm{\bw}_1$ by definition of $\gamma_{\bw}$.
Combined with our bounds on linear portions this shows $\Apot$ and $\Bpot$ are $\text{poly}(nd\rho)$-Lipschitz.

Hence, we may evaluate to the desired $O(\frac \eta n)$ accuracy by approximate minimization of a Lipschitz convex function over an interval (Lemma 33, \cite{CohenLMPS16}) with a total cost of $O(d\log^3(nd\rho))$. Here we use the subroutine of Lemma~\ref{lem:optimize-softmax} in Lemma 33 of \cite{CohenLMPS16}, with evaluation time $O(d\log^2(nd\rho))$. 

The algorithm then runs in $NN'$ iterations, each bottlenecked by the cost of approximating $\Gamma_t$; combining these multiplicative factors yields the runtime. We note that we do not precompute $\Delta = \ma v$; we can compute coordinates of $\Delta$ in time $O(d)$ as they are required by Algorithm~\ref{alg:ostep}.
\end{proof}

\subsection{Equivalence between Assumption~\ref{assume:det} and RIP}\label{ssec:detassume}

The main result of this section is an equivalence between Assumption~\ref{assume:det} and the weighted restricted isometry property, which requires two helper tools to prove. The first is a ``shelling decomposition.'' 

\begin{lemma}\label{lem:shell}
Let $v \in \R^d$ have $\NS(v) \le \sigma$. Then if we write $v = \sum_{l \in [k]} v^{(l)}$ where $v^{(1)}$ is obtained by taking the $s$ largest coordinates of $v$, $v^{(2)}$ is obtained by taking the next $s$ largest coordinates and so on (breaking ties arbitrarily so that the supports are disjoint), we have
\[\sum_{2 \le l \le k} \norm{v^{(l)}}_2 \le  \sqrt{\frac \sigma s} \norm{v}_2.\]
\end{lemma}
\begin{proof}
Note that the decomposition greedily sets $v^{(l)}$ to be the $s$ largest coordinates (by absolute value) of $v - \sum_{l' \in [l - 1]} v^{(l')}$, zeroing all other coordinates and breaking ties arbitrarily. This satisfies
\[\norm{v^{(l + 1)}}_2 \le \sqrt{s} \norm{v^{(l + 1)}}_\infty \le \frac{1}{\sqrt s} \norm{v^{(l)}}_1.\]
The last inequality follows since every entry of $v^{(l)}$ is larger than the largest of $v^{(l + 1)}$ in absolute value. Finally, summing the above equation and using disjointness of supports yields
\begin{align*}
\sum_{2 \le l \le k} \norm{v^{(l)}}_2 \le \frac{1}{\sqrt s} \norm{v}_1 \le  \sqrt{\frac{\sigma}{s}} \norm{v}_2.
\end{align*}

\end{proof}

The second bounds the largest entries of image vectors from the transpose of an RIP matrix.

\begin{lemma}\label{lem:atbound}
Let $\ma \in \R^{n \times d}$ be $(s, c)$-RIP, and let $u \in \R^n$. Then,
\[\norm{\ma^\top u}_{2, (s)} \le c \norm{u}_2.\]
\end{lemma}
\begin{proof}
Let $v \in \R^d$ be the $s$-sparse vector obtained by zeroing out all but the $s$ largest coordinates of $\ma^\top u$. The lemma is equivalent to showing $\norm{v}_2 \le c \norm{u}_2$. Note that
\[\norm{v}_2^2 = \inprod{v}{\ma^\top u} \le \norm{\ma v}_2 \norm{u}_2 \le c \norm{v}_2 \norm{u}_2.\]
The first inequality used Cauchy-Schwarz, and the second applied the RIP property of $\ma$ to $v$, which is $s$-sparse by construction. The conclusion follows via dividing by $\norm{v}_2$.
\end{proof}

Using these helper tools, we now prove the main result of this section.
\begin{lemma}\label{lem:assume==RIP}
	The following statements are true.
	\begin{enumerate}
		\item If $\ma$ satisfies Assumption~\ref{assume:det} with weight vector $w^\star$, then $(\mw^\star)^{\half} \ma$ is $(s, L)$-RIP.
		\item If the matrix $(\mw^\star)^{\half}\ma$ is RIP with parameters
		\[
		\Par{12800 L^3 K^2 s , \frac {\sqrt L} 2}
		\]
		for $L \ge 1$, and $\norm{\ma}_{\max} \le \rho$, then $\ma$ satisfies Assumption~\ref{assume:det}.
	\end{enumerate}
\end{lemma}

\begin{proof} We prove each equivalence in turn.
	
\paragraph{Assumption~\ref{assume:det} implies RIP.} The statement of RIP is scale-invariant, so we will prove it for all $s$-sparse unit vectors $v$ without loss of generality. Note that such $v$ satisfies the condition in Assumption~\ref{assume:det}, since $\norm{v}_2 = 1$ and $\norm{v}_1 \le \sqrt{s}$ by Cauchy-Schwarz. Then, the second condition of Assumption~\ref{assume:det} implies
that for $\Delta = \ma v$, we have the desired norm preservation:
\[
\frac 1 L \leq \norm{(\mw^*)^{\half}\ma v }_2^2 = \sum_{i \in [n]} w^\star_i \Delta_i^2 \leq L.
\]

\paragraph{Boundedness and RIP imply Assumption~\ref{assume:det}.} Let $v \in \R^d$ satisfy $\frac 1 4 \leq \norm{v}_2 \leq 1$ and $\norm{v}_1 \leq 2 \sqrt{2s}$, and define $\Delta \defeq \ma v$. The first condition in Assumption~\ref{assume:det} is immediate from our assumed entrywise boundedness on $\ma$, so we begin by demonstrating the lower bound in \eqref{eq:rsc}. Let
\[
s' =12800 L^3 K^2 s
\]
and let $v^{(1)}, \ldots , v^{(k)}$ be the shelling decomposition into $s'$-sparse vectors given by Lemma~\ref{lem:shell}, where $\sigma = 128s$ from the $\ell_1$ and $\ell_2$ norm bounds on $v$.  By Lemma~\ref{lem:shell}, we have 
\[
\norm{v^{(2)}}_2  + \dots +  \norm{v^{(k)}}_2 \leq \frac{0.1}{L}\norm{v}_2.
\]
In particular, the triangle inequality then implies $0.9 \norm{v}_2 \le \norm{v^{(1)}}_2 \le \norm{v}_2$. Next, recall that $\sum_{i \in [n]} w^\star_i \Delta_i^2 = \norm{(\mw^\star)^{\half} \ma v}_2^2$. By applying the triangle inequality and since $(\mw^\star)^{\half} \ma$ is RIP,
\begin{align*}
\norm{(\mw^\star)^{\half}\ma v }_2^2 &\geq \left(\norm{(\mw^\star)^{\half}\ma v^{(1)} }_2 - \sum_{l = 2}^k  \norm{(\mw^\star)^{\half}\ma v^{(l)}  }_2  \right)^2  \\
&\ge \Par{\frac 5 {\sqrt L} \cdot 0.9 \norm{v}_2 - \frac{\sqrt{L}}{2} \cdot \frac 1 L\norm{v}_2}^2 \geq \frac{16}{L} \norm{v}_2^2 \ge \frac 1 L.
\end{align*}
In the second inequality, we applied the RIP assumption to each individual term, since all the vectors are $s'$-sparse. Similarly, to show the upper bound in \eqref{eq:rsc}, we have
\begin{align*}
\norm{(\mw^\star)^{\half}\ma v }_2^2 &\leq \left(\norm{(\mw^\star)^{\half}\ma v^{(1)} }_2 + \sum_{l = 2}^k  \norm{(\mw^\star)^{\half}\ma v^{(l)}  }_2  \right)^2 \\
&\le \Par{ \frac{\sqrt{L}}{2} \cdot \norm{v}_2 + \frac{\sqrt{L}}{2} \cdot \frac 1 L \norm{v}_2}^2 \le L.
\end{align*}

It remains to verify the final condition of Assumption~\ref{assume:det}. First, for $u \defeq \mw^{\half} \ma v$, by applying the shelling decomposition to $v$ into $s'$-sparse vectors $\{v^{(l)}\}_{l \in [k]}$,
\begin{equation}\label{eq:rip_eq_equiv_1}
\norm{u}_2 \le \sum_{l \in [k]} \norm{(\mw^\star)^{\half} \ma v^{(l)}}_2 \le \sqrt{L} \norm{v}_2.
\end{equation}
Here, we used our earlier proof to bound the contribution of all terms but $v^{(1)}$. Applying Lemma~\ref{lem:atbound} to the matrix $(\mw^\star)^{\half} \ma$ and vector $u$, we have for $\Delta = \ma v$,
\[\norm{\ma^\top (\mw^\star)^{\half} u}_{2, (s')} = \norm{\ma^\top \mw^\star \Delta }_{2, (s')} \le L.\]
By setting the $\ell_2$ bounded component in the short-flat decomposition of $\ma^\top \mw^\star \Delta$ to be the top $s'$ entries by magnitude, it remains to show the remaining coordinates are $\ell_\infty$ bounded by $\frac 1 {K\sqrt{s}}$.  
 This follows from the definition of $s'$ and \eqref{eq:rip_eq_equiv_1}, which imply that the $s' + 1^{\text{th}}$ largest coordinate (in magnitude) cannot have squared value larger than $\frac{L^2}{s'} \le \frac{1}{K^2s}$ without contradicting \eqref{eq:rip_eq_equiv_1}.
\end{proof}

Finally, it is immediate that Lemma~\ref{lem:detundergen} follows from Lemma~\ref{lem:assume==RIP}.

\subsection{Putting it all together}

At this point, we have assembled the tools to prove our main result on exact recovery.

\restateexact*
\begin{proof}
With probability at least $1 - \delta$, combining Lemma~\ref{lem:oracleunderdet} and Lemma~\ref{lem:detundergen} implies that Assumption~\ref{assume:det} holds for all $v \in \R^d$ where $\frac 1 4 \le \norm{v}_2 \le 1$ and $\norm{v}_1 \le 2\sqrt{2s}$, and that for $N = O(\log \frac{R_0} r)$, we can implement a step oracle for $N$ runs of Algorithm~\ref{alg:hrs} in the allotted time, each with failure probability $1 - \frac \delta N$. Moreover, Algorithm~\ref{alg:hrs} returns in $O(1)$ iterations, and allows us to halve our radius upper bound. By taking a union bound on failure probabilities and repeatedly running Algorithm~\ref{alg:hrs} $N$ times, we obtain a radius upper bound of $r$ with probability $\ge 1 - \delta$.
\end{proof} 	%

\section{Noisy recovery}
\label{sec:noisy}

In this section, we give an algorithm for solving a noisy sparse recovery problem in a wRIP matrix $\ma \in \R^{n \times d}$ (where we recall Definition~\ref{def:wrip}). In particular, we assume that we receive
\begin{equation}\label{eq:noiseb}b = \ma \sx + \xi^\star,\end{equation}
for an arbitrary unknown $\xi^\star \in \R^n$, and $\sx \in \R^d$ is $s$-sparse. Throughout this section, we will define 
\begin{equation}\label{eq:mdef}m \defeq \frac{1}{\wsinf},\end{equation}
where $\wsinf$ is an entrywise bound on $w$ in Definition~\ref{def:wrip}. We define the (unknown) ``noise floor''
\[R_\xi \defeq \frac 1 {\sqrt m}\norm{\xi^\star}_{2, (m)},\]
where we defined $\norm{\cdot}_{2, (m)}$ in Section~\ref{sec:prelims}. Our goal will be to return $x$ such that $\norm{x - \sx}_2 = O(R_\xi)$. We now formally state the main result of this section here.

\begin{restatable}{theorem}{restatenoisy}\label{thm:noisy}
	Let $\delta \in (0, 1)$, and suppose $R_0 \ge \norm{\sx}_2$ for $s$-sparse $\sx \in \R^d$. Further, suppose $\ma \in \R^{n \times d}$ is $(\rho, \wsinf)$-wRIP and $b = \ma \sx + \xi^\star$, and $R_1 \ge R_\xi \defeq \frac 1 {\sqrt m}\norm{\xi^\star}_{2, (m)}$. Then with probability at least $1 - \delta$, Algorithm~\ref{alg:hrsn} using Algorithm~\ref{alg:hrsn} as a noisy step oracle computes $\hat{x}$ satisfying
	\[\norm{\hat x - \sx}_2 \le R_{\textup{final}} = \Theta(R_{\xi}), \]
	in time 
	\begin{align*}
		O\Par{\Par{ nd \wsinf  s \log^4 (nd\rho)\log^2\Par{\frac{d}{\delta} \cdot \log\Par{\frac{R_0}{R_{\textup{final}}}}\log\Par{\frac{R_1}{R_{\textup{final}}}} }} \cdot \rho^2 \log\Par{\frac{R_0}{R_{\textup{final}}}}  \log\Par{\frac{R_1}{R_{\textup{final}}}}}.
	\end{align*}
\end{restatable}

Similarly to Theorem~\ref{thm:exact}, Theorem~\ref{thm:noisy} provides a runtime guarantee which interpolates between the fully random and semi-random settings, and runs in sublinear time when e.g.\ the entire measurement matrix $\ma$ satisfies RIP. Theorem~\ref{thm:noisy} further provides a refined error guarantee as a function of the noise vector $\xi$, which again interpolates based on the ``quality'' of the weights $w$. This is captured through the parameter $m = \frac 1 {\wsinf}$: when $m \approx n$, the squared error bound $R_\xi^2$ scales as the average squared entry of $\xi$, and more generally it scales as the average of the largest $m$ entries.

We solve the noisy variant by essentially following the same steps as Section~\ref{sec:exact} and making minor modifications to the analysis; we give an outline of the section here.  In Section~\ref{ssec:onephasenoisy}, we generalize the framework of Section~\ref{ssec:onephase} to the setting where we only receive noisy observations \eqref{eq:noiseb}, while our current radius is substantially above the noise floor. We then implement an appropriate step oracle for this outer loop in Section~\ref{ssec:oraclenoisy}, and prove that the relevant Assumption~\ref{assume:detplus} used in our step oracle implementation holds when $\ma$ is wRIP in Section~\ref{ssec:detassumenoisy}.

\subsection{Radius contraction above the noise floor using step oracles}\label{ssec:onephasenoisy}

In this section, we give the main loop of our overall noise-tolerant algorithm, $\HRSN$, which takes as input $s$-sparse $\xin$ and a radius bound $R \ge \norm{\xin - \sx}_2$. It then returns an $s$-sparse vector $\xout$ with the guarantee $\norm{\xout - \sx}_2 \le \half R$, \emph{as long as} $R$ is larger than an appropriate multiple of $R_\xi$. We give the analog of Definition~\ref{def:ostep} in this setting, termed a ``noisy step oracle.''

\begin{definition}[Noisy step oracle]\label{def:onstep}
We say that $\onstep$ is a $(\Cp, C_2, C_\xi, \delta)$-noisy step oracle for $\tDelta \in \R^n$ and $\ma \in \R^{n \times d}$ if the following holds. Whenever there is $v \in \R^d$ with $\frac 1 {12} \le \norm{v}_2 \le 1$ such that $\tDelta = \ma v + \xi$ where $\norm{\xi}_{2, (m)} \le \frac{\sqrt m}{C_\xi}$, with probability $\ge 1- \delta$, $\onstep$ returns $w \in \R^n_{\ge 0}$ such that the following two conditions hold. First,
\begin{equation}\label{eq:tdeltadelta}\sum_{i \in [n]} w_i \tDelta_i\Delta_i \ge \Cp.\end{equation}
Second, there exists a $(C_2, \frac{\Cp}{6 \sqrt s})$ short-flat decomposition of $\ma^\top \diag{w} \Delta$:
\[\norm{\trunc\Par{\ma^\top \diag{w} \Delta, \frac{\Cp}{6 \sqrt s}}}_2 \le C_2.\]
\end{definition}

We next characterize how a strengthened step oracle with appropriate parameters also is a noisy step oracle.  First, we will need a definition.

\begin{definition}\label{def:stochastic-domination}
For distributions $A,B$ on $\R^n$, we say $A$ stochastically dominates $B$ if there is a random variable $C$ on $\R^n$ whose coordinates are always nonnegative such that the distribution of $A$ is the same as the distribution of $B + C$ (where $C$ may depend on the realization of $B$).
\end{definition}

We now formalize the properties of the strengthened step oracle that we will construct.

\begin{definition}[Strong step oracle]\label{def:osstep}
We say that $\ostep$ is a $(\Cp, C_2, C_\xi, \delta)$-strong step oracle for $\tDelta \in \R^n$ and $\ma \in \R^{n \times d}$ if it satisfies all the properties of a standard step oracle (Definition~\ref{def:ostep}), as well as the following additional guarantees.
\begin{enumerate}
	\item For the output weights $w$, we have
	\begin{equation}\label{eq:wbounds}\norm{w}_1 \le \frac{\Cp C_\xi^2}{4} \cdot \delta.\end{equation}
	\item The distribution of $w$ output by the oracle is stochastically dominated by the distribution 
	\[
	 \frac{\delta}{4s \rho^2  \log \frac d \delta} \textup{Multinom}\left( \left\lceil \frac{\Cp C_{\xi}^2 n s  \rho^2 \log \frac d \delta}{m} \right\rceil , \left( \underbrace{\frac{1}{n}, \dots , \frac{1}{n}}_n \right ) \right) 
	\]
	for some $\rho \geq 1$.
	
	\item Compared to Definition~\ref{def:ostep} (the step oracle definition), we have the stronger guarantees that $\ma^\top \diag{w} \Delta$ admits a $(C_2, \frac{\Cp}{24\sqrt{s}})$ short-flat decomposition in \eqref{eq:shortflat}, and obtains its guarantees using the bounds $\frac 1 {12} \le \norm{v}_2 \le 1$ (instead of a lower bound of $\frac 1 4$).
\end{enumerate}
\end{definition}

We next demonstrate that a strong step oracle is a noisy step oracle.
\begin{lemma}\label{lem:ostepisonstep}
Suppose $\ostep$ is a $(\Cp, C_2, C_\xi, \delta)$-strong step oracle for $\tDelta \in \R^n$ and $\ma \in \R^{n \times d}$. Then, $\ostep$ is also a $(\frac 1 4 \Cp, C_2, C_\xi, 2\delta)$-noisy step oracle for $(\tDelta, \ma)$.
\end{lemma}
\begin{proof}
In the definition of a noisy step oracle, we only need to check the condition that $\sum_{i \in [n]} w_i \tDelta_i \Delta_i \ge \frac 1 4 \Cp$ for an arbitrary $\Delta = \tDelta - \xi$ where $\norm{\xi}_{2, (m)} \le \sqrt m C_\xi^{-1}$, as all other conditions are immediate from Definition~\ref{def:osstep}. Note that
\begin{align*}
\sum_{i \in [n]} w_i \tDelta_i\Delta_i &= \sum_{i \in [n]} w_i \tDelta_i(\tDelta_i - \xi_i) \\
&\ge \half \sum_{i \in [n]} w_i \tDelta_i^2 - \half \sum_{i \in [n]} w_i \xi_i^2 \,.
\end{align*}
where we used $a^2 - ab \ge \half a^2 - \half b^2$.  The first sum above is at least $\half\Cp $ by assumption.  To upper bound the second sum, we will use the second property in the definition of a strong step oracle.  Let $S \subset [n]$ be the set consisting of the $m$ largest coordinates of $\xi$ (with ties broken lexicographically).  Let 
$\alpha$ be drawn from the distribution
\[
 \frac{\delta }{4 s \rho^2 \log \frac d \delta} \textup{Multinom}\left( \left\lceil \frac{\Cp C_{\xi}^2 n s \rho^2 \log \frac d \delta}{m} \right\rceil , \left( \underbrace{\frac{1}{n}, \dots , \frac{1}{n}}_n \right ) \right) .
\]
Note that with $1 - 0.1 \delta$ probability, by a Chernoff bound, we have that $\sum_{i \in S} \alpha_i \geq \frac 1 5 \delta \Cp C_{\xi}^2$.  If this happens, then since $S$ consists of the largest coordinates of $\xi$, any vector $\beta$ such that $\beta \leq \alpha$ entrywise and $\norm{\beta}_1 \leq  \frac 1 4 \delta \Cp C_{\xi}^2$ must have 
\[
\sum_{i \in [n]} \beta_i \xi_i^2 \leq \frac{5}{4} \sum_{i \in S} \alpha_i \xi_i^2.
\]
Now note that for any $S$ with $|S| = m$,
\[
\E\Brack{ \sum_{i \in S} \alpha_i \xi_i^2 } \leq  \frac{\delta \Cp C_{\xi}^2 }{4 m } \cdot \norm{\xi}_{2, (m)}^2 \leq \frac{\delta \Cp}{4}. 
\]
Combining the above two inequalities and Markov's inequality and the fact that the distribution of $\alpha$ stochastically dominates the distribution of $w$, we deduce that with at least $1 - \delta$ probability,
\[
\sum_{i \in [n]} w_i \xi_i^2 \leq \frac{1}{0.9 \delta} \cdot \E \left[ \max_{\substack{\beta \leq \alpha \\ \norm{\beta}_1 \leq \frac 1 4 \delta \Cp C_{\xi}^2}} \sum_{i \in [n]} \beta_i \xi_i^2 \right]  \leq \frac{  \Cp}{2}. 
\]
Putting everything together, we conclude that we have 
\[
\sum_{i \in [n]} w_i \tDelta_i\Delta_i \geq \frac{\Cp}{4}
\]
with failure probability at most $2\delta$, completing the proof.
\end{proof}

In Section~\ref{ssec:oraclenoisy}, we prove that if $\ma$ satisfies Assumption~\ref{assume:detplus} (a slightly different assumption than Assumption~\ref{assume:det}) then with high probability we can implement a strong step oracle with appropriate parameters. This is stated more formally in the following; recall $m$ is defined in \eqref{eq:mdef}.

\begin{restatable}{assumption}{restatedetplus}\label{assume:detplus}
	The matrix $\ma \in \R^{n \times d}$ satisfies the following. There is a weight vector $w^\star \in \Delta^n$ with $\norm{w^\star}_\infty \le \wsinf = \frac 1 m$, a constants $L$, $\rho \ge 1$, and constants $K, C_\xi$ (which may depend on $L$) such that for all $v \in \R^d$, $\xi \in \R^n$ with
	\[\frac 1 4 \le \norm{v}_2 \le 1,\; \norm{v}_1 \le 2\sqrt{2s},\; \norm{\xi}_{2, (m)} \le \frac{\sqrt m}{C_\xi}\]
	we have, defining $\tDelta = \ma v + \xi$:
	\begin{enumerate}
		\item $\ma$ is entrywise bounded by $\pm \rho$, i.e.\ $\norm{\ma}_{\max} \le \rho$.
		\item
		\begin{equation}\label{eq:rscplus}\frac 1 L \le \sum_{i \in [n]} w_i^\star \tDelta_i^2 \le L.\end{equation}
		\item There is a $(L, \frac 1 {K \sqrt s})$ short-flat decomposition of $\ma^\top \mw^\star \tDelta$:
		\begin{equation}\label{eq:shortflatvplus}\norm{\trunc\Par{\ma^\top \mw^\star \tDelta, \frac 1 {K \sqrt s} }}_2 \le L.\end{equation}
	\end{enumerate}	
\end{restatable}

\begin{restatable}{lemma}{restateoracleunderdetplus}\label{lem:oracleunderdetplus}
Suppose $\ma$ satisfies Assumption~\ref{assume:detplus}. Algorithm~\ref{alg:onstep} is a $(\Cp, C_2, C_\xi, \delta)$ strong step oracle $\ONStep$ for $(\tDelta, \ma)$ with 
\[\Cp = \Omega(1),\; C_2 = O\Par{1},\; C_\xi = O(1), \; \delta = \half\Par{\frac{C_2}{10^5\Cp}}^2   ,  \]
running in time
\[O\Par{\Par{ nd\log^3 (nd\rho)\log \frac 1 \delta} \cdot \Par{\wsinf  s\rho^2\log^2 \frac{d}{\delta}}}.\]
\end{restatable}

Here, in contrast to the noiseless setting, we can only guarantee that the strong step oracle (and thus also the noisy step oracle) succeeds with constant probability.  In our full algorithm, we boost the success probability of the oracle by running a logarithmic number of independent trials and aggregating the outputs. We also show that for an appropriate choice of constants in Definition~\ref{def:wrip}, Assumption~\ref{assume:detplus} is also satisfied, stated in Lemma~\ref{lem:detundergennoise} and proven in Section~\ref{ssec:detassumenoisy}.

\begin{restatable}{lemma}{restatedetundergennoise}\label{lem:detundergennoise}
Suppose $\ma \in \R^{n \times d}$ is $(\rho, \wsinf)$-wRIP with a suitable choice of constants in the RIP parameters in Definition~\ref{def:wrip}. Then, $\ma$ also satisfies Assumption~\ref{assume:detplus}.
\end{restatable}

\begin{algorithm2e}[ht!]
	\caption{$\HRSN(x_{\text{in}}, R, R_\xi, \onstep, \delta, \ma, b)$}
	\label{alg:hrsn}
		\codeInput $s$-sparse $\xin \in \R^d$, $R \ge \norm{\xin - \sx}_2$ for $s$-sparse $\sx \in \R^d$, $(\Cp, C_2, C_\xi, \delta')$-noisy step oracle $\onstep$ for all $(\Delta, \ma)$ with $\Delta \in \R^n$, $\delta'  \leq (10^{-4} \frac{\Cp}{C_2})^2$, $\ma \in \R^{n \times d}$, $b = \ma \sx + \xi^\star \in \R^n$ for $\norm{\xi^\star}_{2, (m)} \le R_\xi \sqrt m$, with $R \ge C_\xi R_\xi$ \;
		\codeOutput $s$-sparse vector $\xout$ that satisfies  $\norm{x_{\text{out}} - \sx}_2 \le \half R$ with probability $\ge 1 - \delta$
		\codeLineSace
		
		$x_0 \gets x_{\text{in}}$, $\xset \gets \{x \in \R^d \mid \norm{x - \xin}_1 \le \sqrt{2s} R\}$ \;
		$T \gets \left\lceil \frac{100C_2^2}{\Cp^2} \right\rceil$, $\eta \gets \frac{\Cp}{2C_2^2}$ \;
		$N_{\text{trials}} \gets 10\log \frac d \delta$ \;
		\For{$1 \leq j \leq N_{\textup{trials}}$ }{
			$x_0^j \gets x_0$ \;
			\For{$0 \le t \le T - 1$}{
				$w_t^{j}   \gets \onstep(\Delta_t^{j}, \ma)$ for $\Delta_t^{j} \gets \frac{1}{R}(\ma x_t^j - b)$, $ \gamma_t^{j} \gets \ma^\top \diag{[w_t^{j}]} \Delta_t^j = \sum_{i \in [n]} [w_t^j]_i [\Delta_t^j]_i a_i$ \;
				\If{$(w_t^{j}, \gamma_t^{j})$ do not meet all of \eqref{eq:prog}, \eqref{eq:shortflat} and the additional criteria in Definition~\ref{def:osstep}}{
					$x_T^j \gets x_t^j$ truncated to its $s$ largest coordinates \;
					\codeBreak \;
				}
				$x_{t + 1}^j \gets \argmin_{x \in \xset} \norm{x - x_t^j - \eta R \gamma_t}_2$ \;
			}
		}
		$ x_T \gets \textsf{Aggregate}(\{x_T^{1}, \dots , x_T^{N_{\textup{trials}}} \}, \frac R 2 )$ \;
		\codeReturn $\xout \gets x_T$ truncated to its $s$ largest coordinates \;
\end{algorithm2e}

\begin{algorithm2e}[ht!]
	\caption{$\textsf{Aggregate}(\mathcal{S}, R)$}
	\label{alg:agg}
	\codeInput $\mathcal{S} = \{y_i\}_{i \in [k]} \subset \R^d$, $R \ge 0$ such that for some unknown $z \in \R^d$, at least $0.51k$ points $y_i \in \mathcal{S}$ have $\norm{y_i - z}_2 \leq \frac R 3$ \;
	\codeOutput $\wt{z}$ with $\norm{\wt{z} - z}_2 \leq R$ \;
	\For{$1 \leq i \leq k$ }{
		\lIf{at least $0.51k$ points $y_j \in \mathcal{S}$ satisfy $\norm{y_i - y_j}_2 \leq \frac{2R}{3}$}{
			\codeReturn $\wt{z} \leftarrow y_i$
		}
	}
\end{algorithm2e}

Next, we give a guarantee regarding our geometric post-processing step, Algorithm~\ref{alg:agg}.

\begin{claim}\label{claim:agg}
$\mathsf{Aggregate}(\mathcal{S},R)$ runs in $O(k^2d)$ time and meets its output guarantees. 
\end{claim}
\begin{proof}
Let $T$ be the subset of indices $i \in [k]$ such that $\norm{y_i - z} \leq \frac R 3$.  Whenever the algorithm tests $y_i$ for some $i \in T$, it will be returned and satisfies the desired properties.  Now consider any $y_i$ returned by the algorithm. The ball of radius $\frac{2R}{3}$ around $y_i$ intersects the ball of radius $\frac R 3$ around $z$, since otherwise it can only contain at most $0.49k$ points.  Thus, $\norm{y_i  - z }_2 \leq R$. The runtime is dominated by the time it takes to do $k^2$ distance comparisons of points in $\R^d$.
\end{proof}

We remark that is possible that for $k = \Omega(\log \frac 1 \delta)$ as is the case in our applications, the runtime of Claim~\ref{claim:agg} can be improved to have a better dependence on $k$ by subsampling the points and using low-rank projections for distance comparisons.

\begin{lemma}\label{lem:hrsncorrect}
Assume $\ma$ satisfies Assumption~\ref{assume:detplus}. Then, Algorithm~\ref{alg:hrsn} meets its output guarantees in time
\[O\Par{\Par{ nd\log^3 (nd\rho)} \cdot \Par{\wsinf  s\rho^2\log d} \cdot \log^2 \frac d \delta}.\]
\end{lemma}
\begin{proof}
We claim that for each independent trial $j \in [N_{\text{trials}}]$, except with probability $1 - T\delta'$, the output $x_T^j$ satisfies $\|x_T^j - x^\star\|_2 \le \frac R 6$. Once we prove this, by Chernoff at least $0.51 N_{\text{trials}}$ of the trials satisfy $\|x_T^j - x^\star\|_2 \le \frac R 6$ except with probability at most $\delta$, and then we are done by Claim~\ref{claim:agg}.
 
It remains to prove the above claim. Fix a trial $j$, and drop the superscript $j$ for notational convenience. In every iteration $t$, $\tDelta \defeq \frac 1 R (\ma x_t - b)$ is given to $\onstep$. Since $b = \ma \sx + \xi^\star$, we have
\[\tDelta = \frac 1 R (\ma(x - \sx) + \xi^\star) = \ma v + \xi,\]
for $\norm{v}_2 \le 1$, $\norm{v}_1 \le 2\sqrt{2s}$, and $\norm{\xi}_{2, (m)} \le \frac{\sqrt m}{C_\xi}$, where the last inequality used the assumed bounds 
\[\norm{\xi^\star}_{2, (m)} \le R_\xi \sqrt m,\; \frac{R_\xi}{R} \le \frac 1 {C_\xi}.\]
Hence, by the assumptions on $\onstep$, it will not fail for such inputs unless $\norm{v}_2 \ge \frac 1 {12}$ is violated, except with probability $\le \delta'$. If the check in Line 10 fails, then except with probability $\le \delta'$, the conclusion $\|x_T - x^\star\|_2 \le \frac R 6$ follows analogously to Lemma~\ref{lem:combinealg}, since $v = \frac 1 R (x - x^\star)$.

The other case's correctness follows identically to the proof of Lemma~\ref{lem:combinealg}, except for one difference: to lower bound the progress term \eqref{eq:progressterm}, we use the assumption \eqref{eq:tdeltadelta} which shows
\[2\eta R \inprod{\gamma_t}{x_t - \sx} = 2\eta R \sum_{i \in [n]} w_i \tDelta_i \inprod{a_i}{v} = 2\eta R^2 \sum_{i \in [n]} w_i \tDelta_i \Delta_i \ge 2\eta R^2 \Cp.\]
Hence, following the proof of Lemma~\ref{lem:combinealg} (and adjusting for constants), whenever the algorithm does not terminate we make at least a $\frac{50}{T}$ fraction of the progress towards $x^\star$, so in $T$ iterations (assuming no step oracle failed) we will have $\|x_T - x^\star\|_2 \le \frac R 6$.

Finally, the runtime follows from combining Lemma~\ref{lem:oracleunderdetplus} (with constant failure probability) with a multiplicative overhead of $T \cdot N_{\text{trials}}$ due to the number of calls to the step oracle, contributing one additional logarithmic factor. We adjusted one of the $\log d$ terms to become a $\log \frac d \delta$ term to account for the runtime of $\mathsf{Aggregate}$ (see Claim~\ref{claim:agg}).
\end{proof}

\subsection{Designing a strong step oracle}\label{ssec:oraclenoisy}

In this section, we design a strong step oracle $\ostep(\tDelta, \ma)$ under Assumption~\ref{assume:detplus}. As in Section~\ref{ssec:oracle}, our oracle iteratively builds a weight vector $\bw$, and sets
\[\gamma_{\bw} \defeq \sum_{i \in [n]} \bw_i \tDelta_i a_i.\]
We will use essentially the same potentials as in \eqref{eq:abdef}, defined in the following:
\begin{equation}\label{eq:abtdef}
\tApot(\bw) \defeq \sum_{i \in [n]} \bw_i \tDelta_i^2,\; \tBpot(\bw) \defeq \Par{\min_{\norm{p}_2 \le  L\norm{\bw}_1} \sqmax_\mu(\gamma_{\bw} - p)} + \frac{ \norm{\bw}_1}{4CLs}.
\end{equation}

\begin{algorithm2e}[ht!]
	\caption{$\ONStep(\tDelta, \ma, \delta)$}
	\label{alg:onstep}
		\codeInput $\tDelta \in \R^n, \ma \in \R^{n \times d}$ satisfying Assumption~\ref{assume:det}, $\delta \in (0, 1)$ \;
		\codeOutput $(w, \gamma)$ such that $\gamma = \sum_{i \in [n]} w_i \tDelta_i a_i$, and if there is $v \in \R^d$ with $\frac 1 {12} \le \norm{v}_2 \le 1$ such that $\tDelta = \ma v + \xi$ where $\norm{\xi}_{2, (m)} \le \frac{\sqrt m}{C_\xi}$, with probability $\ge 1 - \delta$, \eqref{eq:prog}, \eqref{eq:shortflat} are satisfied with
		\[\Cp = 1,\; C_2 = O\Par{1}.\]
		Furthermore, the second condition in \eqref{eq:shortflat} is satisfied with the constant $24$ rather than $6$, and there is $C_\xi = O(1)$ such that \eqref{eq:wbounds} is also satisfied.
		\codeLineSace
		
		$C \gets 3200$, $\mu \gets \frac{1}{\sqrt{Cs \log d}}$, $\eta \gets \frac 1 {K\wsinf s\rho^2\log d}$, $N' \gets \lceil \log_2 \frac 2 \delta \rceil$ \;
		\For{$0 \le k \le N'$}{
			$w_0 \gets 0_n$, $N \gets \lceil \frac {5Ln} {\eta}\rceil$ \;
			\For{$0 \le t \le N$}{
			\lIf{$\tApot(w_t) \ge 1$}{
				\codeReturn $\gamma \gets \sum_{i \in [n]} [w_t]_i \tDelta_i a_i$, $w \gets w_t$
			}
			Sample $i \sim_{\textup{unif.}} [n]$ \;
			Compute (using Lemma~\ref{lem:optimize-softmax}) $d_t \in [0, \eta \wsinf]$ maximizing to additive $O(\frac \eta n )$
			\[\Gamma_t(d) \defeq \tApot(w_t + d e_i) - C s \tBpot(w_t + d e_i)\]
			$w_{t + 1} \gets w_t + d_t e_i$ \;
			}
		}
		\codeReturn $\gamma \gets 0_d$, $w \gets 0_n$
	\end{algorithm2e}

Algorithm~\ref{alg:onstep} is essentially identical to Algorithm~\ref{alg:ostep} except for changes in constants.
We further have the following which verifies the second property in the definition of strong step oracle.
\begin{fact}\label{fact:winfbound}
The distribution of $w$ returned by Algorithm~\ref{alg:onstep} is stochastically dominated by the distribution
\[
\eta w_{\infty}^* \textup{Multinom}\left(  \frac{5Ln}{\eta} , \left( \underbrace{\frac{1}{n}, \dots , \frac{1}{n}}_n \right)  \right)
\]
\end{fact}
\begin{proof}
Every time we inspect a row, we change the corresponding entry of $w$ by at most $\eta w_{\infty}^*$. The result follows from the number of iterations in the algorithm and uniformity of sampling rows.
\end{proof}

To analyze Algorithm~\ref{alg:onstep}, we provide appropriate analogs of Lemmas~\ref{lem:bgrowth} and~\ref{lem:oracleunderdet}. Because Algorithm~\ref{alg:onstep} is very similar to Algorithm~\ref{alg:ostep}, we will largely omit the proof of the following statement, which follows essentially identically to the proof of Lemma~\ref{lem:bgrowth} up to adjusting constants.

\begin{lemma}\label{lem:bgrowthn}
Assume that the constant $K$ in Assumption~\ref{assume:detplus} is sufficiently large, and that $\tDelta = \ma v + \xi$ where $v, \xi$ satisfy the norm conditions in Assumption~\ref{assume:detplus}. Then for any $\bw \in \R^n_{\ge 0}$ such that $B(\bw) \le C^2 \mu^2 \log d$, we have
\[\E_{i \sim_{\textup{unif.}} [n]} [B(\bw + \eta w_i^\star)] \le B(\bw) + \frac{1}{2CLs} \cdot \frac{\eta}{n}.\]
\end{lemma}
\begin{proof}
The analysis is essentially identical to that of Lemma~\ref{lem:bgrowth}; we discuss only the main difference. To apply the Taylor expansion of the exponential, Lemma~\ref{lem:bgrowth} required a bound that
\[\frac 1 \mu \Abs{z_j^{(i)}} = O\Par{\frac 1 {\sqrt{\log d}}} \impliedby \Abs{z_j^{(i)}} = O\Par{\frac 1 {\sqrt s \log d}}, \]
for all $i \in [n]$ and $j \in [d]$. Note that in the setting of Lemma~\ref{lem:bgrowth}, we took $z_j^{(i)} = \eta w_i^\star (\Delta_i a_{ij} - p^\star_{j})$. Here, we will take $z_j^{(i)} = \eta w_i^\star (\tDelta_i a_{ij} - p^\star_{j})$; bounds on all of these terms follow identically to in Lemma~\ref{lem:bgrowth}, except that $\tDelta_i = \inprod{a_i}{v} + \xi_i$, so we need to show
\[\eta w_i^\star \rho |\xi_i| = O\Par{\frac 1 {\sqrt{s}\log d}}.\]
This follows since $\eta \le \frac 1 {\log d}$ and $|\xi_i| = O(\sqrt{m})$ by assumption. Hence, as $w_i^\star \le \frac 1 m$ by definition of $m$, this is equivalent to $m = \Omega(s\rho^2)$, an explicit assumption we make. 
\end{proof}

We now give a full analysis of Algorithm~\ref{alg:onstep}, patterned off of Lemma~\ref{lem:oracleunderdet}.

\restateoracleunderdetplus*
\begin{proof}
The analysis is essentially identical to that of Algorithm~\ref{alg:ostep} in Lemma~\ref{lem:oracleunderdet}; we discuss differences here. First, the stochastic domination condition follows from Fact~\ref{fact:winfbound} for sufficiently large $C_{\xi}$, $K$.

For the remaining properties, since the algorithm runs $N' \ge \log_2 \frac 2 \delta$ times independently, it suffices to show each run meets Definition~\ref{def:osstep} with probability $\ge \half$ under the events of Assumption~\ref{assume:detplus}, assuming there exists the desired decomposition $\tDelta = \ma v + \xi$ in the sense of Assumption~\ref{assume:detplus}. Union bounding with the failure probability in Fact~\ref{fact:winfbound} yields the overall failure probability.

\paragraph{Correctness.} As in Lemma~\ref{lem:oracleunderdet}, it is straightforward to see that $\tApot$ is $1$-Lipschitz, since the value of $\eta$ is smaller than that used in Algorithm~\ref{alg:ostep}. The termination condition in iteration $T$ then again implies $\tApot(w_T) \ge 1$, and $\tBpot(w_T) \le \frac 3 {Cs}$. For $C = 3200$, this implies the short-flat decomposition with stronger parameters required by Definition~\ref{def:osstep}, as well as the $\norm{w_T}_1$ bound. 

\paragraph{Success probability.} As in Lemma~\ref{lem:oracleunderdet}, the expected growth in $\Phi_t$ in any iteration where $\Pr[\tApot(w_t) \ge 1] \le \half$ is $\ge \frac {\eta} {4Ln}$. Hence, running for $\ge \frac{5Ln}{\eta}$ iterations and using $\Phi_t - \Phi_0 \le 2$ yields the claim.

\paragraph{Runtime.} This follows identically to the analysis in Lemma~\ref{lem:oracleunderdet}. 

\end{proof}

\subsection{Deterministic assumptions for noisy regression}\label{ssec:detassumenoisy}

In this section, we prove Lemma~\ref{lem:detundergennoise}, restated here for completeness. The proof will build heavily on our previous developments in the noiseless case, as shown in Section~\ref{ssec:detassume}.

\restatedetundergennoise*
\begin{proof}
The analysis is largely similar to the analysis of Lemma~\ref{lem:assume==RIP}; we will now discuss the differences here, which are introduced by the presence of the noise term $\xi$. There are three components to discuss: the upper and lower bounds in \eqref{eq:rscplus}, and the decomposition \eqref{eq:shortflatvplus}.

Regarding the bounds in \eqref{eq:rsc}, by changing constants appropriately in Definition~\ref{def:wrip}, we can assume that $\ma$ satisfies the second property in Assumption~\ref{assume:det} with the parameters $\frac 4 L$ and $\frac L 4$. In particular, for $\Delta = \ma v$, we then have
\[\frac 4 L \le \sum_{i \in [n]} w_i^\star \Delta_i^2 \le \frac L 4.\]
Recall that $\tDelta = \Delta + \xi$ for some $\norm{\xi}_{2, (m)} \le \frac{\sqrt m}{C_\xi}$. Hence,
\begin{align*}
\sum_{i \in [n]} w_i^\star \tDelta_i^2 &\le 2\sum_{i \in [n]} w_i^\star \Delta_i^2 + 2\sum_{i \in [n]} w_i^\star \xi_i^2 \\
&\le \frac L 2 + 2\Par{\frac 1 m \norm{\xi}_{2, (m)}^2} \le L,
\end{align*}
for an appropriately large $C_\xi^2 \ge \frac 4 L$. Here the first inequality used $(a + b)^2 \le 2a^2 + 2b^2$, and the second inequality used that the largest $\sum_{i \in [n]} w_i^\star \xi_i^2$ can be subject to $\norm{w^\star}_1 = 1$ and $\norm{w^\star}_\infty \le \frac 1 m$ is attained by greedily choosing the $m$ largest coordinates of $\xi$ by their magnitude, and setting $w_i^\star = \frac 1 m$ for those coordinates. This gives the upper bound in Assumption~\ref{assume:detplus}, and the lower bound follows similarly: for appropriately large $C_\xi^2 \ge \frac L 2$,
\begin{align*}
\sum_{i \in [n]} w_i^\star \tDelta_i^2 &\ge \half \sum_{i \in [n]} w_i^\star \Delta_i^2 - \half \sum_{i \in [n]} w_i^\star \xi_i^2 \\
&\ge \frac 2 L - \half \Par{\frac 1 m \norm{\xi}_{2, (m)}^2} \ge \frac 1 L.
\end{align*}
Lastly, for the decomposition required by \eqref{eq:shortflatvplus}, we will use the decomposition of Lemma~\ref{lem:detundergen} for the component due to $\sum_{i \in [n]} w_i^\star \Delta_i a_i$; in particular, assume by adjusting constants that this component has a $(\frac L 2, \frac{1}{2K\sqrt s})$ short-flat decomposition. It remains to show that
\[\sum_{i \in [n]} w_i^\star \xi_i a_i = \ma^\top \mw^\star \xi.\]
also admits a $(\frac L 2, \frac{1}{2K\sqrt s})$ short-flat decomposition, at which point we may conclude by the triangle inequality. Let $u = (\mw^\star)^{\half} \xi$; from earlier, we bounded 
\[\norm{u}_2^2 \le \frac 1 m \norm{\xi}_{2, (m)}^2 \implies \norm{u}_2 \le \frac 1 {C_{\xi}}. \]
Hence, applying Lemma~\ref{lem:atbound} using the RIP matrix $(\mw^\star)^{\half} \ma$ with appropriate parameters yields the conclusion, for large enough $C_{\xi}$. In particular, the $\ell_2$-bounded part of the decomposition follows from Lemma~\ref{lem:atbound}, and the proof of the $\ell_\infty$-bounded part is identical to the proof in Lemma~\ref{lem:assume==RIP}.
\end{proof}

\subsection{Putting it all together}

We now prove our main result on noisy recovery.

\restatenoisy*
\begin{proof}
Our algorithm will iteratively maintain a guess $R_{\text{guess}}$ on the value of $\frac 1 {\sqrt m}\norm{\xi^\star}_{2, (m)}$, initialized at $R_{\text{guess}} \gets R_1$. For each value of $R_{\text{guess}} \ge R_\xi$, the hypothesis of Algorithm~\ref{alg:hrsn} is satisfied, and hence using a strategy similar to the proof of Theorem~\ref{thm:exact} (but terminating at accuracy $R = O(R_{\text{guess}})$ where the constant is large enough to satisfy the assumption $R \ge C_\xi R_{\text{guess}}$) results in an estimate at distance $R$ with probability at least $1 - \delta$, with runtime
\[
O\Par{\Par{ nd\wsinf  s\rho^2\log^4 (nd\rho)\log^2\Par{\frac{d}{\delta} \cdot \log\Par{\frac{R_0}{R_{\textup{final}}}} }} \cdot \rho^2 \log\Par{\frac{R_0}{R_{\textup{final}}}} }.
\]
The runtime above follows from Lemma~\ref{lem:hrsncorrect}.

Our overall algorithm repeatedly halves $R_{\text{guess}}$, and outputs the last point returned by a run of the algorithm where it can certify a distance bound to $\sx$ of $R = C_\xi R_{\text{guess}}$. We use $R_{\text{final}}$ to denote $C_\xi R_{\text{guess}}$ on the last run. Clearly for any $R_{\text{guess}} \ge R_\xi$ this certification will succeed, so we at most lose a factor of $2$ in the error guarantee as we will have $R_{\text{final}} \le 2C_\xi R_\xi$. The final runtime follows from adjusting $\delta$ by a factor of $O(\log \frac{R_1}{R_{\text{final}}})$ to account for the multiple runs of the algorithm.
\end{proof} 	
	\newpage
	\bibliographystyle{alpha}	

\newcommand{\etalchar}[1]{$^{#1}$}
\begin{thebibliography}{vdBLL{\etalchar{+}}21}

\bibitem[Ans60]{anscombe1960rejection}
Frank~J Anscombe.
\newblock Rejection of outliers.
\newblock {\em Technometrics}, 2(2):123--146, 1960.

\bibitem[ANW10]{AgarwalNW10}
Alekh Agarwal, Sahand~N. Negahban, and Martin~J. Wainwright.
\newblock Fast global convergence rates of gradient methods for
  high-dimensional statistical recovery.
\newblock In {\em Advances in Neural Information Processing Systems 23: 24th
  Annual Conference on Neural Information Processing Systems 2010. Proceedings
  of a meeting held 6-9 December 2010, Vancouver, British Columbia, Canada},
  pages 37--45, 2010.

\bibitem[ANW12]{agarwal2012fast}
Alekh Agarwal, Sahand Negahban, and Martin~J Wainwright.
\newblock Fast global convergence of gradient methods for high-dimensional
  statistical recovery.
\newblock {\em The Annals of Statistics}, pages 2452--2482, 2012.

\bibitem[AP17]{aich2017application}
Abhishek Aich and P~Palanisamy.
\newblock On application of omp and cosamp algorithms for doa estimation
  problem.
\newblock In {\em 2017 International Conference on Communication and Signal
  Processing (ICCSP)}, pages 1983--1987. IEEE, 2017.

\bibitem[AV18]{AwasthiV18}
Pranjal Awasthi and Aravindan Vijayaraghavan.
\newblock Towards learning sparsely used dictionaries with arbitrary supports.
\newblock In Mikkel Thorup, editor, {\em 59th {IEEE} Annual Symposium on
  Foundations of Computer Science, {FOCS} 2018, Paris, France, October 7-9,
  2018}, pages 283--296. {IEEE} Computer Society, 2018.

\bibitem[BBC11]{becker2011nesta}
Stephen Becker, J{\'e}r{\^o}me Bobin, and Emmanuel~J Cand{\`e}s.
\newblock Nesta: A fast and accurate first-order method for sparse recovery.
\newblock {\em SIAM Journal on Imaging Sciences}, 4(1):1--39, 2011.

\bibitem[BCDH10]{baraniuk2010model}
Richard~G Baraniuk, Volkan Cevher, Marco~F Duarte, and Chinmay Hegde.
\newblock Model-based compressive sensing.
\newblock {\em IEEE Transactions on information theory}, 56(4):1982--2001,
  2010.

\bibitem[BD09]{blumensath2009iterative}
Thomas Blumensath and Mike~E Davies.
\newblock Iterative hard thresholding for compressed sensing.
\newblock {\em Applied and computational harmonic analysis}, 27(3):265--274,
  2009.

\bibitem[BD10]{blumensath2010normalized}
Thomas Blumensath and Mike~E Davies.
\newblock Normalized iterative hard thresholding: Guaranteed stability and
  performance.
\newblock {\em IEEE Journal of selected topics in signal processing},
  4(2):298--309, 2010.

\bibitem[BDMS13]{bandeira2013certifying}
Afonso~S Bandeira, Edgar Dobriban, Dustin~G Mixon, and William~F Sawin.
\newblock Certifying the restricted isometry property is hard.
\newblock {\em IEEE transactions on information theory}, 59(6):3448--3450,
  2013.

\bibitem[Bel18]{bellec2018noise}
Pierre~C Bellec.
\newblock The noise barrier and the large signal bias of the lasso and other
  convex estimators.
\newblock {\em arXiv preprint arXiv:1804.01230}, 2018.

\bibitem[Blu03]{blum2003machine}
Avrim Blum.
\newblock Machine learning: My favorite results, directions, and open problems.
\newblock In {\em 44th Annual IEEE Symposium on Foundations of Computer
  Science, 2003. Proceedings.}, pages 2--2. IEEE, 2003.

\bibitem[BRT09]{bickel2009simultaneous}
Peter~J Bickel, Ya’acov Ritov, and Alexandre~B Tsybakov.
\newblock Simultaneous analysis of lasso and dantzig selector.
\newblock {\em The Annals of statistics}, 37(4):1705--1732, 2009.

\bibitem[BRW21]{bhaskara2021principal}
Aditya Bhaskara, Aravinda~Kanchana Ruwanpathirana, and Maheshakya Wijewardena.
\newblock Principal component regression with semirandom observations via
  matrix completion.
\newblock In {\em International Conference on Artificial Intelligence and
  Statistics}, pages 2665--2673. PMLR, 2021.

\bibitem[BS95]{blum1995coloring}
Avrim Blum and Joel Spencer.
\newblock Coloring random and semi-random k-colorable graphs.
\newblock {\em Journal of Algorithms}, 19(2):204--234, 1995.

\bibitem[BT09]{beck2009fast}
Amir Beck and Marc Teboulle.
\newblock A fast iterative shrinkage-thresholding algorithm for linear inverse
  problems.
\newblock {\em SIAM journal on imaging sciences}, 2(1):183--202, 2009.

\bibitem[CG18]{cheng2018non}
Yu~Cheng and Rong Ge.
\newblock Non-convex matrix completion against a semi-random adversary.
\newblock In {\em Conference On Learning Theory}, pages 1362--1394. PMLR, 2018.

\bibitem[CJSX14]{chen2014clustering}
Yudong Chen, Ali Jalali, Sujay Sanghavi, and Huan Xu.
\newblock Clustering partially observed graphs via convex optimization.
\newblock {\em The Journal of Machine Learning Research}, 15(1):2213--2238,
  2014.

\bibitem[CKMY20]{chen2020classification}
Sitan Chen, Frederic Koehler, Ankur Moitra, and Morris Yau.
\newblock Classification under misspecification: Halfspaces, generalized linear
  models, and connections to evolvability.
\newblock {\em arXiv preprint arXiv:2006.04787}, 2020.

\bibitem[CLM{\etalchar{+}}16]{CohenLMPS16}
Michael~B. Cohen, Yin~Tat Lee, Gary~L. Miller, Jakub Pachocki, and Aaron
  Sidford.
\newblock Geometric median in nearly linear time.
\newblock In Daniel Wichs and Yishay Mansour, editors, {\em Proceedings of the
  48th Annual {ACM} {SIGACT} Symposium on Theory of Computing, {STOC} 2016,
  Cambridge, MA, USA, June 18-21, 2016}, pages 9--21. {ACM}, 2016.

\bibitem[CR05]{candes2005signal}
Emmanuel~J Candes and Justin~K Romberg.
\newblock Signal recovery from random projections.
\newblock In {\em Computational Imaging III}, volume 5674, pages 76--86.
  International Society for Optics and Photonics, 2005.

\bibitem[CRT06]{candes2006stable}
Emmanuel~J Candes, Justin~K Romberg, and Terence Tao.
\newblock Stable signal recovery from incomplete and inaccurate measurements.
\newblock {\em Communications on Pure and Applied Mathematics: A Journal Issued
  by the Courant Institute of Mathematical Sciences}, 59(8):1207--1223, 2006.

\bibitem[CSX12]{chen2012clustering}
Yudong Chen, Sujay Sanghavi, and Huan Xu.
\newblock Clustering sparse graphs.
\newblock {\em arXiv preprint arXiv:1210.3335}, 2(5), 2012.

\bibitem[CT06]{candes2006near}
Emmanuel~J Candes and Terence Tao.
\newblock Near-optimal signal recovery from random projections: Universal
  encoding strategies?
\newblock {\em IEEE transactions on information theory}, 52(12):5406--5425,
  2006.

\bibitem[CW05]{combettes2005signal}
Patrick~L Combettes and Val{\'e}rie~R Wajs.
\newblock Signal recovery by proximal forward-backward splitting.
\newblock {\em Multiscale Modeling \& Simulation}, 4(4):1168--1200, 2005.

\bibitem[DDDM04]{daubechies2004iterative}
Ingrid Daubechies, Michel Defrise, and Christine De~Mol.
\newblock An iterative thresholding algorithm for linear inverse problems with
  a sparsity constraint.
\newblock {\em Communications on Pure and Applied Mathematics: A Journal Issued
  by the Courant Institute of Mathematical Sciences}, 57(11):1413--1457, 2004.

\bibitem[DDEK12]{davenport2012introduction}
Mark~A Davenport, Marco~F Duarte, Yonina~C Eldar, and Gitta Kutyniok.
\newblock Introduction to compressed sensing., 2012.

\bibitem[DGT19]{diakonikolas2019distribution}
Ilias Diakonikolas, Themis Gouleakis, and Christos Tzamos.
\newblock Distribution-independent pac learning of halfspaces with massart
  noise.
\newblock {\em arXiv preprint arXiv:1906.10075}, 2019.

\bibitem[DHL17]{dalalyan2017prediction}
Arnak~S Dalalyan, Mohamed Hebiri, and Johannes Lederer.
\newblock On the prediction performance of the lasso.
\newblock {\em Bernoulli}, 23(1):552--581, 2017.

\bibitem[DIK{\etalchar{+}}21]{diakonikolas2021boosting}
Ilias Diakonikolas, Russell Impagliazzo, Daniel Kane, Rex Lei, Jessica Sorrell,
  and Christos Tzamos.
\newblock Boosting in the presence of massart noise.
\newblock {\em arXiv preprint arXiv:2106.07779}, 2021.

\bibitem[DK20]{diakonikolas2020hardness}
Ilias Diakonikolas and Daniel~M Kane.
\newblock Hardness of learning halfspaces with massart noise.
\newblock {\em arXiv preprint arXiv:2012.09720}, 2020.

\bibitem[DKK{\etalchar{+}}21]{diakonikolas2021threshold}
Ilias Diakonikolas, Daniel~M Kane, Vasilis Kontonis, Christos Tzamos, and Nikos
  Zarifis.
\newblock Threshold phenomena in learning halfspaces with massart noise.
\newblock {\em arXiv preprint arXiv:2108.08767}, 2021.

\bibitem[DKT21]{diakonikolas2021forster}
Ilias Diakonikolas, Daniel~M Kane, and Christos Tzamos.
\newblock Forster decomposition and learning halfspaces with noise.
\newblock {\em arXiv preprint arXiv:2107.05582}, 2021.

\bibitem[DKTZ20]{diakonikolas2020learning}
Ilias Diakonikolas, Vasilis Kontonis, Christos Tzamos, and Nikos Zarifis.
\newblock Learning halfspaces with massart noise under structured
  distributions.
\newblock In {\em Conference on Learning Theory}, pages 1486--1513. PMLR, 2020.

\bibitem[DNW13]{davenport2013signal}
Mark~A Davenport, Deanna Needell, and Michael~B Wakin.
\newblock Signal space cosamp for sparse recovery with redundant dictionaries.
\newblock {\em IEEE Transactions on Information Theory}, 59(10):6820--6829,
  2013.

\bibitem[Don06]{donoho2006compressed}
David~L Donoho.
\newblock Compressed sensing.
\newblock {\em IEEE Transactions on information theory}, 52(4):1289--1306,
  2006.

\bibitem[DPT21]{diakonikolas2021relu}
Ilias Diakonikolas, Jongho Park, and Christos Tzamos.
\newblock Relu regression with massart noise.
\newblock {\em arXiv preprint arXiv:2109.04623}, 2021.

\bibitem[DS89]{donoho1989uncertainty}
David~L Donoho and Philip~B Stark.
\newblock Uncertainty principles and signal recovery.
\newblock {\em SIAM Journal on Applied Mathematics}, 49(3):906--931, 1989.

\bibitem[EK12]{eldar2012compressed}
Yonina~C Eldar and Gitta Kutyniok.
\newblock {\em Compressed sensing: theory and applications}.
\newblock Cambridge university press, 2012.

\bibitem[ES09]{elsner2009bounding}
Micha Elsner and Warren Schudy.
\newblock Bounding and comparing methods for correlation clustering beyond ilp.
\newblock In {\em Proceedings of the Workshop on Integer Linear Programming for
  Natural Language Processing}, pages 19--27, 2009.

\bibitem[FK00]{feige2000finding}
Uriel Feige and Robert Krauthgamer.
\newblock Finding and certifying a large hidden clique in a semirandom graph.
\newblock {\em Random Structures \& Algorithms}, 16(2):195--208, 2000.

\bibitem[FK01]{feige2001heuristics}
Uriel Feige and Joe Kilian.
\newblock Heuristics for semirandom graph problems.
\newblock {\em Journal of Computer and System Sciences}, 63(4):639--671, 2001.

\bibitem[FN03]{figueiredo2003algorithm}
M{\'a}rio~AT Figueiredo and Robert~D Nowak.
\newblock An em algorithm for wavelet-based image restoration.
\newblock {\em IEEE Transactions on Image Processing}, 12(8):906--916, 2003.

\bibitem[Fou11]{foucart2011hard}
Simon Foucart.
\newblock Hard thresholding pursuit: an algorithm for compressive sensing.
\newblock {\em SIAM Journal on Numerical Analysis}, 49(6):2543--2563, 2011.

\bibitem[GLS88]{GLS1988}
Martin Gr{\"{o}}tschel, L{\'{a}}szl{\'{o}} Lov{\'{a}}sz, and Alexander
  Schrijver.
\newblock {\em Geometric Algorithms and Combinatorial Optimization}, volume~2
  of {\em Algorithms and Combinatorics}.
\newblock Springer, 1988.

\bibitem[GRSY14]{globerson2014tight}
Amir Globerson, Tim Roughgarden, David Sontag, and Cafer Yildirim.
\newblock Tight error bounds for structured prediction.
\newblock {\em arXiv preprint arXiv:1409.5834}, 2014.

\bibitem[Hub64]{huber1964robust}
Peter~J Huber.
\newblock Robust estimation of a location parameter.
\newblock {\em The Annals of Mathematical Statistics}, pages 73--101, 1964.

\bibitem[JTK14]{jain2014iterative}
Prateek Jain, Ambuj Tewari, and Purushottam Kar.
\newblock On iterative hard thresholding methods for high-dimensional
  m-estimation.
\newblock In {\em NIPS}, 2014.

\bibitem[KKMR21]{kelner2021power}
Jonathan Kelner, Frederic Koehler, Raghu Meka, and Dhruv Rohatgi.
\newblock On the power of preconditioning in sparse linear regression.
\newblock {\em arXiv preprint arXiv:2106.09207}, 2021.

\bibitem[KM14]{koltchinskii2014l1}
Vladimir Koltchinskii and Stanislav Minsker.
\newblock $l_1$-penalization in functional linear regression with subgaussian
  design.
\newblock {\em Journal de l’Ecole polytechnique-Math{\'e}matiques},
  1:269--330, 2014.

\bibitem[KMM11]{kolla2011play}
Alexandra Kolla, Konstantin Makarychev, and Yury Makarychev.
\newblock How to play unique games against a semi-random adversary: Study of
  semi-random models of unique games.
\newblock In {\em 2011 IEEE 52nd Annual Symposium on Foundations of Computer
  Science}, pages 443--452. IEEE, 2011.

\bibitem[Kut13]{kutyniok2013theory}
Gitta Kutyniok.
\newblock Theory and applications of compressed sensing.
\newblock {\em GAMM-Mitteilungen}, 36(1):79--101, 2013.

\bibitem[LF81]{levy1981reconstruction}
Shlomo Levy and Peter~K Fullagar.
\newblock Reconstruction of a sparse spike train from a portion of its spectrum
  and application to high-resolution deconvolution.
\newblock {\em Geophysics}, 46(9):1235--1243, 1981.

\bibitem[LSTZ20]{li2020wellconditioned}
Jerry Li, Aaron Sidford, Kevin Tian, and Huishuai Zhang.
\newblock Well-conditioned methods for ill-conditioned systems: Linear
  regression with semi-random noise, 2020.

\bibitem[MD10]{maleki2010optimally}
Arian Maleki and David~L Donoho.
\newblock Optimally tuned iterative reconstruction algorithms for compressed
  sensing.
\newblock {\em IEEE Journal of Selected Topics in Signal Processing},
  4(2):330--341, 2010.

\bibitem[MMV12]{makarychev2012approximation}
Konstantin Makarychev, Yury Makarychev, and Aravindan Vijayaraghavan.
\newblock Approximation algorithms for semi-random partitioning problems.
\newblock In {\em Proceedings of the forty-fourth annual ACM symposium on
  Theory of computing}, pages 367--384, 2012.

\bibitem[MMV13]{makarychev2013sorting}
Konstantin Makarychev, Yury Makarychev, and Aravindan Vijayaraghavan.
\newblock Sorting noisy data with partial information.
\newblock In {\em Proceedings of the 4th conference on Innovations in
  Theoretical Computer Science}, pages 515--528, 2013.

\bibitem[MMV14]{makarychev2014constant}
Konstantin Makarychev, Yury Makarychev, and Aravindan Vijayaraghavan.
\newblock Constant factor approximation for balanced cut in the pie model.
\newblock In {\em Proceedings of the forty-sixth annual ACM symposium on Theory
  of computing}, pages 41--49, 2014.

\bibitem[MMV15]{makarychev2015correlation}
Konstantin Makarychev, Yury Makarychev, and Aravindan Vijayaraghavan.
\newblock Correlation clustering with noisy partial information.
\newblock In {\em Conference on Learning Theory}, pages 1321--1342. PMLR, 2015.

\bibitem[MMV16]{makarychev2016learning}
Konstantin Makarychev, Yury Makarychev, and Aravindan Vijayaraghavan.
\newblock Learning communities in the presence of errors.
\newblock In {\em Conference on learning theory}, pages 1258--1291. PMLR, 2016.

\bibitem[MN06]{massart2006risk}
Pascal Massart and {\'E}lodie N{\'e}d{\'e}lec.
\newblock Risk bounds for statistical learning.
\newblock {\em The Annals of Statistics}, 34(5):2326--2366, 2006.

\bibitem[Moi17]{moitra2017robustness}
Ankur Moitra.
\newblock What does robustness say about algorithms.
\newblock ICML '17 Tutorial, 2017.

\bibitem[MPW16]{moitra2016robust}
Ankur Moitra, William Perry, and Alexander~S Wein.
\newblock How robust are reconstruction thresholds for community detection?
\newblock In {\em Proceedings of the forty-eighth annual ACM symposium on
  Theory of Computing}, pages 828--841, 2016.

\bibitem[MS10]{mathieu2010correlation}
Claire Mathieu and Warren Schudy.
\newblock Correlation clustering with noisy input.
\newblock In {\em Proceedings of the twenty-first annual ACM-SIAM symposium on
  Discrete Algorithms}, pages 712--728. SIAM, 2010.

\bibitem[MZ93]{mallat1993matching}
St{\'e}phane~G Mallat and Zhifeng Zhang.
\newblock Matching pursuits with time-frequency dictionaries.
\newblock {\em IEEE Transactions on signal processing}, 41(12):3397--3415,
  1993.

\bibitem[NRWY12]{negahban2012unified}
Sahand~N Negahban, Pradeep Ravikumar, Martin~J Wainwright, and Bin Yu.
\newblock A unified framework for high-dimensional analysis of $ m $-estimators
  with decomposable regularizers.
\newblock {\em Statistical science}, 27(4):538--557, 2012.

\bibitem[NT09]{needell2009cosamp}
Deanna Needell and Joel~A Tropp.
\newblock Cosamp: Iterative signal recovery from incomplete and inaccurate
  samples.
\newblock {\em Applied and computational harmonic analysis}, 26(3):301--321,
  2009.

\bibitem[NV10]{needell2010signal}
Deanna Needell and Roman Vershynin.
\newblock Signal recovery from incomplete and inaccurate measurements via
  regularized orthogonal matching pursuit.
\newblock {\em IEEE Journal of selected topics in signal processing},
  4(2):310--316, 2010.

\bibitem[PCBVB14]{polania2014exploiting}
Luisa~F Polania, Rafael~E Carrillo, Manuel Blanco-Velasco, and Kenneth~E
  Barner.
\newblock Exploiting prior knowledge in compressed sensing wireless ecg
  systems.
\newblock {\em IEEE journal of Biomedical and Health Informatics},
  19(2):508--519, 2014.

\bibitem[PRK93]{pati1993orthogonal}
Yagyensh~Chandra Pati, Ramin Rezaiifar, and Perinkulam~Sambamurthy
  Krishnaprasad.
\newblock Orthogonal matching pursuit: Recursive function approximation with
  applications to wavelet decomposition.
\newblock In {\em Proceedings of 27th Asilomar conference on signals, systems
  and computers}, pages 40--44. IEEE, 1993.

\bibitem[Rou21]{roughgarden2021beyond}
Tim Roughgarden.
\newblock {\em Beyond the Worst-Case Analysis of Algorithms}.
\newblock Cambridge University Press, 2021.

\bibitem[RV06]{rudelson2006sparse}
Mark Rudelson and Roman Vershynin.
\newblock Sparse reconstruction by convex relaxation: Fourier and gaussian
  measurements.
\newblock In {\em 2006 40th Annual Conference on Information Sciences and
  Systems}, pages 207--212. IEEE, 2006.

\bibitem[RWY10]{raskutti2010restricted}
Garvesh Raskutti, Martin~J Wainwright, and Bin Yu.
\newblock Restricted eigenvalue properties for correlated gaussian designs.
\newblock {\em The Journal of Machine Learning Research}, 11:2241--2259, 2010.

\bibitem[Sch18]{ludwig2018algorithms}
Ludwig Schmidt.
\newblock {\em Algorithms above the noise floor}.
\newblock PhD thesis, Massachusetts Institute of Technology, 2018.

\bibitem[SS86]{santosa1986linear}
Fadil Santosa and William~W Symes.
\newblock Linear inversion of band-limited reflection seismograms.
\newblock {\em SIAM Journal on Scientific and Statistical Computing},
  7(4):1307--1330, 1986.

\bibitem[TG07]{tropp2007signal}
Joel~A Tropp and Anna~C Gilbert.
\newblock Signal recovery from random measurements via orthogonal matching
  pursuit.
\newblock {\em IEEE Transactions on information theory}, 53(12):4655--4666,
  2007.

\bibitem[Tuk60]{tukey1960survey}
John~W Tukey.
\newblock A survey of sampling from contaminated distributions.
\newblock {\em Contributions to probability and statistics}, pages 448--485,
  1960.

\bibitem[Tuk75]{tukey1975mathematics}
John~W Tukey.
\newblock Mathematics and the picturing of data.
\newblock In {\em Proceedings of the International Congress of Mathematicians,
  Vancouver, 1975}, volume~2, pages 523--531, 1975.

\bibitem[vdBLL{\etalchar{+}}21]{BrandLLSS0W21}
Jan van~den Brand, Yin~Tat Lee, Yang~P. Liu, Thatchaphol Saranurak, Aaron
  Sidford, Zhao Song, and Di~Wang.
\newblock Minimum cost flows, mdps, and
  {\(\mathscr{l}\)}\({}_{\mbox{1}}\)-regression in nearly linear time for dense
  instances.
\newblock In {\em {STOC} '21: 53rd Annual {ACM} {SIGACT} Symposium on Theory of
  Computing, Virtual Event, Italy, June 21-25, 2021}, pages 859--869, 2021.

\bibitem[vdBLSS20]{BrandLSS20}
Jan van~den Brand, Yin~Tat Lee, Aaron Sidford, and Zhao Song.
\newblock Solving tall dense linear programs in nearly linear time.
\newblock In {\em Proccedings of the 52nd Annual {ACM} {SIGACT} Symposium on
  Theory of Computing, {STOC} 2020, Chicago, IL, USA, June 22-26, 2020}, pages
  775--788, 2020.

\bibitem[VdGL13]{van2013lasso}
Sara Van~de Geer and Johannes Lederer.
\newblock The lasso, correlated design, and improved oracle inequalities.
\newblock In {\em From Probability to Statistics and Back: High-Dimensional
  Models and Processes--A Festschrift in Honor of Jon A. Wellner}, pages
  303--316. Institute of Mathematical Statistics, 2013.

\bibitem[ZL21]{zhang2021improved}
Chicheng Zhang and Yinan Li.
\newblock Improved algorithms for efficient active learning halfspaces with
  massart and tsybakov noise.
\newblock {\em arXiv preprint arXiv:2102.05312}, 2021.

\bibitem[ZWJ17]{zhang2017optimal}
Yuchen Zhang, Martin~J Wainwright, and Michael~I Jordan.
\newblock Optimal prediction for sparse linear models? lower bounds for
  coordinate-separable m-estimators.
\newblock {\em Electronic Journal of Statistics}, 11(1):752--799, 2017.

\bibitem[ZWW{\etalchar{+}}16]{zhang2016comparison}
Zhimin Zhang, Shoushui Wei, Dingwen Wei, Liping Li, Feng Liu, and Chengyu Liu.
\newblock Comparison of four recovery algorithms used in compressed sensing for
  ecg signal processing.
\newblock In {\em 2016 Computing in Cardiology Conference (CinC)}, pages
  401--404. IEEE, 2016.

\end{thebibliography}
\newcommand{\etalchar}[1]{$^{#1}$}

	\newpage
	\begin{appendix}

\section{Greedy and non-convex methods fail in the semi-random setting} \label{app:failure}

In this section, we show how a few standard, commonly-used non-convex or greedy methods can fail (potentially quite drastically) in the semi-random adversary setting.  The two algorithms that we examine are Iterative Hard Thresholding and Orthogonal Matching Pursuit \cite{blumensath2009iterative, tropp2007signal}. We believe it is likely that similar counterexamples can be constructed for other, more complex algorithms such as CoSaMP \cite{needell2009cosamp}. For simplicity in this section, we will only discuss the specific semi-random model introduced in Definition~\ref{def:prip}, where $\ma$ is pRIP, i.e.\ it contains an unknown RIP matrix $\mg$ as a subset of its rows.

\subsection{Iterative hard thresholding}

The iterative hard thresholding algorithm \cite{blumensath2009iterative} involves initializing $x_0 = 0$ and taking 
\[
x_{t+1} = H_s \left( x_t - \frac{1}{n} \ma^\top (b - \ma x_t  ) \right) 
\]
where $H_s$ zeroes out all but the $s$ largest entries in magnitude (ties broken lexicographically).  We can break this algorithm in the semi-random setting by simply duplicating one row many times.

\paragraph{Hard semi-random adversary.}  Let $n = Cm$ for some sufficiently large constant $C$.  The first $m$ rows of $\ma$ are drawn independently from $\Nor(0, \id)$.  Now draw $v \sim \Nor(0,\id)$, except set the first entry of $v$ to $1$.   We set the last $(C - 1)m $ rows of $\ma$ all equal to $v$.  We will set the sparsity parameter $s = 1$ and let $\sx = (1 , 0, \ldots , 0)$.  We let $b = \ma \sx$.  
\begin{proposition}
With $\ma,b$ generated as above, with high probability, iterative hard thresholding does not converge.
\end{proposition}
\begin{proof}
With high probability, some coordinate of $v$ is $\Omega(\sqrt{\log d})$.  We then have that some entry of $\ma^\top b$ has magnitude at least $\Omega(n \sqrt{\log d}  )$ with high probability. Thus, the next iterate $x_{1}$ must have exactly one nonzero entry that has magnitude at least $\Omega(\sqrt{\log d})$ and furthermore, this entry must correspond to some coordinate of $v$ that has magnitude at least $\Omega(\sqrt{\log d})$.  However, this means that the residuals in all of the rows that are copies of $v$ are at least $\Omega(\log d )$.  In the next step, by the same argument, we get that the residuals blow up even more and clearly this algorithm will never converge.  In fact, $x_t$ will never have the right support because its support will always be on one of the entries where $v$ is large.
\end{proof}

\subsection{Orthogonal matching pursuit}
The orthogonal matching pursuit algorithm \cite{tropp2007signal} involves initializing $x_0 = 0$ and keeping track of a set $S$ (that corresponds to our guess of the support of $\sx$).  Each iteration, we choose a column $c_j$ of $\ma$ that maximizes $\frac{| \langle c_j , r_t \rangle |}{\norm{c_j}_2^2}$ and then add $j$ to $S$ (where $r_t = \ma x_t - b$ is the residual).  We then add $j$ to $S$ and project the residual onto the orthogonal complement of all coordinates in $S$.  We show that we can again very easily break this algorithm in the semi-random setting.

\paragraph{Hard semi-random adversary.} Let $n = 3m$.  First, we draw all rows of $\ma$ independently from $\Nor(0,\id)$.  Next, we modify some of the entries in the last $2m$ rows.  Let $s$ be the sparsity parameter. Let $\sx = (s^{-\half}, \ldots , s^{-\half}, 0, \ldots, 0 )$ be supported on the first $s$ coordinates and set $b = \ma \sx$.  Now we modify the columns of $\ma$ (aside from the first $s$ so $\ma \sx$ is not affected). We set the last $2m$ entries of one of these columns $c_j$ to match those of $b$.
\begin{proposition}
With $\ma, b$ generated as above, with high probability, orthogonal matching pursuit does not recover $\sx$.
\end{proposition}
\begin{proof}
With high probability (as long as $s \geq 10$), the column $c_j$ is the one that maximizes  $\frac{| \langle c_j , b \rangle |}{\norm{c_j}_2^2}$ because its last $2m$ entries exactly match those of $b$.  However, $j$ is not in the support of $\sx$ so the algorithm has already failed.       
\end{proof}

We further make the following observation.
\begin{remark}
By modifying other columns of $\ma$ as well, the semi-random adversary can actually make the algorithm pick all of the wrong columns in the support.
\end{remark}

\subsection{Convex methods} 

Now we briefly comment on how convex methods are robust, in the sense that they can still be used in the semi-random setting (but may have substantially slower rates than their fast counterparts).  In the noiseless observations case, this is clear because the additional rows of $\ma$ are simply additional constraints that are added to the standard $\ell_1$ minimization convex program.  

In the noisy case, let the target error be $\theta = \norm{\xi^*}_{2, (m)}$.  We then solve the modified problem 
\begin{align*}
&\min \norm{x}_1 \\
&\text{subject to } \norm{\ma x - b}_{2,(m)} \leq \theta .
\end{align*}
Note that the above is a convex program and thus can be solved in polynomial time by e.g.\ cutting plane methods \cite{GLS1988}.  Also, note that $\sx$ is indeed feasible for the second constraint.  Now for the solution $\widehat{x}$ that we obtain, we must have $\norm{\widehat{x}}_1 \leq \norm{\sx}_1$ and 
\[
\norm{\ma(\sx - \widehat{x})}_{2,(m)} \leq 2\theta .
\]
Let $\mg$ be the set of $m$ randomly generated rows of $\ma$ under our semi-random adversarial model.  The previous two conditions imply
\begin{itemize}
    \item $\norm{\widehat{x} - \sx}_1 \leq 2\sqrt{s} \norm{\sx - \widehat{x}}_2$
    \item $\norm{\mg(\sx -  \widehat{x})}_{2} \leq 2\theta $
\end{itemize}
which now by restricted strong convexity of $\mg$ (see \cite{AgarwalNW10}) implies that $\norm{\sx -  \widehat{x}}_2 = O( \frac{\theta}{\sqrt{m}})$.  We can furthermore round $\widehat{x}$ to $s$-sparse to obtain the sparse vector $x'$, and the above bound only worsens by a factor of $2$ for $x'$ (see Lemma~\ref{lem:combinealg} for this argument).

\section{Deferred proofs}\label{app:deferred}

\restatesqmax*
\begin{proof}
Let $\mathcal{P} \subset \R^d$ be the set of $p$ such that $p$ has the same sign as $\gamma$ entrywise and $|p_j| \le |\gamma_j|$ for all $j \in [d]$. By symmetry of the $\sqmax$ and the $\ell_2$ norm under negation, the optimal $p$ lies in $\mathcal{P}$.
	
Next we claim that the function $\smax_\mu(\gamma - p)$ is $2\norm{\gamma}_2$-Lipschitz in the $\ell_2$ norm as a function of $p$, over $\mathcal{P}$. To see this, the gradient is directly computable as 
\[2(p - \gamma) \circ x
\text{ where }
x \in  \Delta^d \text{ with } x_i = \frac{\exp([\gamma_i - p_i]^2 / \mu^2)}{\sum_{j \in [n]} \exp([\gamma_j - p_j]^2 / \mu^2)}
\text{ for all } i \in [n]
\]
where  $\circ$ denotes entrywise multiplication. Thus, the $\ell_2$ norm of the derivative is bounded by $2\norm{\gamma}_2$ over $\mathcal{P}$. In the remainder of the proof, we show how to find $p \in \mathcal{P}$ which has $\ell_2$ error $\frac{\delta}{2\norm{\gamma}_2}$ to the optimal, which implies by Lipschitzness that the function value is within additive $\delta$ of optimal.

Next, since $0 \in \mathcal{P}$, we may assume without loss of generality that
	\begin{equation}\label{eq:param-lowerbound}
	\theta > \frac{\delta}{ 2 \norm{\gamma}_2}.
	\end{equation}
else we may just output $0$, which achieves optimality gap at most $2\norm{\gamma}_2 \theta$.
	
	Now, by monotonicity of $\ln$ it suffices to approximately minimize
	\[
	\sum_{j \in [d]} \exp \left( \frac{[\gamma - p]_j^2 }{\mu^2} \right).
	\]
	The sum above is always at least $d$.   First we check if $\norm{\gamma}_2 \leq \theta + \sqrt{\delta}$.  If this is true then clearly we can set $p$ so that all entries of $\gamma - p$ have magnitude at most $\sqrt{\delta}$.  This gives a solution such that
	\[
	\sqmax_\mu(\gamma - p) \leq \mu^2 \log \left( d \exp\left(\frac{\delta}{\mu^2} \right) \right) = \mu^2 \log d + \delta
	\]
	and since the value of $\sqmax$ is always at least $\mu^2 \log d$, this solution is optimal up to additive error $\delta$.  Thus, we can assume $\norm{\gamma}_2 \geq \theta + \sqrt{\delta} $  in the remainder of the proof. We also assume all entries of $\gamma$ are nonzero since if an entry of $\gamma$ is $0$ then the corresponding entry of $p$ should also be $0$. Finally by symmetry of the problem under negation we will assume all entries of $\gamma$ are positive in the remainder of the proof, such that each entry of $p$ is also positive.
	
	By monotonicity of $\sqmax$ in each coordinate (as long as signs are preserved) and the assumption that $\norm{\gamma}_2 \geq \theta + \sqrt{\delta}$, the optimal solution must have $\norm{p}_2 = \theta$. By using  Lagrange multipliers, for some scalar $\zeta$ and all $j$,
	\begin{equation}\label{eq:pjdef}
	p_j = \exp(\zeta) \cdot [\gamma - p]_j  \exp \left( \frac{[\gamma - p]_j^2}{\mu^2} \right).
	\end{equation}
	For the optimal $\zeta$ by taking $\ell_2$ norms of the quantity above, we have
	\[\theta = \norm{p}_2 = \zeta\norm{\gamma - p}_2 \cdot C \text{ for some } C \in \Brack{0, \exp\Par{\frac{\norm{\gamma}_2^2}{\mu^2}}}.\]
	Hence taking logarithms of both sides and using both the bounds \eqref{eq:param-lowerbound} and $\norm{\gamma - p}_2 \ge \sqrt{\delta}$ at the optimum, which follows from the previous discussion, we obtain
	\[\log \frac{\theta}{\norm{\gamma - p}_2} - \zeta \in \Brack{0, \frac{\norm{\gamma}_2^2}{\mu^2}} \implies \zeta \in \Brack{-\frac{\norm{\gamma}_2^2}{\mu^2} - \log\Par{\frac{2\norm{\gamma}_2^2}{\delta}}, \log\Par{\frac{\norm{\gamma}_2}{\sqrt{\delta}}}}.\]
	We next show how to compute $p$ to high accuracy given a guess on $\zeta$. Observe that if $\gamma_j > 0$, then the right-hand side of \eqref{eq:pjdef} is decreasing in $p_j$ and hence by the intermediate value theorem, there is a unique solution strictly between $0$ and $\gamma_j$ for any $\zeta$. Also, note that the location of this solution increases with $\zeta$. Let $p(\zeta)$ be the solution obtained by exactly solving \eqref{eq:pjdef} for some given $\zeta$. We have shown for all $\zeta$ that $0 \le [p(\zeta)]_j \le \gamma_j$ entrywise and hence $\norm{p(\zeta)}_2 \le \norm{\gamma}_2$ for all $\zeta$. 
	
	For a fixed $\zeta$, we claim we can estimate $p(\zeta)$ to $\ell_2$ error $\beta$ in time $O(d\log \frac{\norm{\gamma}_2}{\beta})$. To see this, fix some $\zeta$, $\mu$, and $\gamma_j$, and consider solving \eqref{eq:pjdef} for the fixed point $p_j$. We can discretize $[0, \gamma_j]$ into intervals of length $\frac{\gamma_j \beta}{\norm{\gamma}_2}$ and perform a binary search. The right-hand side is decreasing in $p_j$ and the left-hand side is increasing so the binary search yields some interval of length $\frac{\gamma_j \beta}{\norm{\gamma}_2}$ containing the fixed point $p_j$ via the intermediate value theorem. The resulting $\ell_2$ error along all coordinates is then $\beta$. 	We also round this approximate $p(\zeta)$ entrywise down in the above search to form a vector $\tp(\zeta, \beta)$ such that $\tp(\zeta, \beta) \le p(\zeta)$ entrywise and $\norm{\tp(\zeta, \beta) - p(\zeta)}_2 \le \beta$. We use this notation and it is well-defined as the search is deterministic. 
	
	In the remainder of the proof we choose the constants
	\[\alpha \defeq \frac{\delta^2}{192\norm{\gamma}_2^4},\; \beta \defeq \min\Par{\frac{\delta^2}{192\norm{\gamma}_2^3}, \frac{\delta}{4\norm{\gamma}_2}}.\]
	We define $\tp(\zeta) \defeq \tp(\zeta, \beta)$ for short as $\beta$ will be fixed. Discretize the range $[-\frac{\norm{\gamma}_2^2}{\mu^2} - \log \frac{2\norm{\gamma}_2^2}{\delta}, \log \frac{\norm{\gamma}_2}{\sqrt \delta}]$ into a grid of uniform intervals of length $\alpha$. Consider the $\zeta$ such that $\zeta \le \zeta^\star < \zeta + \alpha$. Because $p(\zeta^\star)$ is entrywise larger than $p(\zeta)$ and hence the logarithmic term on the right-hand side of \eqref{eq:pjdef} is smaller for $p(\zeta^\star)$ than $p(\zeta)$, we have
	\[ \Brack{p(\zeta)}_j \le \Brack{p\Par{\zeta^\star}}_j \le  \exp\Par{\alpha} \Brack{p\Par{\zeta}}_j.\]
	Moreover the optimal $p(\zeta^\star)$ has $\ell_2$ norm $\theta$, so $|\zeta - \zeta^\star| \le \alpha$ and $\exp(\alpha) - 1 \le 2\alpha$ imply
	\[\norm{p(\zeta) - p(\zeta^\star)}_2 \le 2\alpha \norm{p(\zeta^\star)}_2 \le 2\alpha\norm{\gamma}_2 \le \Delta \defeq \frac{\delta^2}{96\norm{\gamma}_2^3}.\]
	Consider the algorithm which returns the $\zalg$ on the search grid which minimizes $|\|\tp(\zalg)\|_2 - \theta|$ (we will discuss computational issues at the end of the proof). As we have argued above, there is a choice which yields $\norm{p(\zeta)}_2 \in [\theta - \Delta, \theta + \Delta]$ and hence
	\begin{equation}\label{eq:tprange}\norm{\tp(\zalg)}_2 \in \Brack{\theta - \Delta - \beta, \theta + \Delta + \beta}.\end{equation}
	We next claim that 
	\begin{equation}\label{eq:closep}
	\norm{p(\zalg) - p(\zeta^\star)}_2 \le \frac{\delta}{4\norm{\gamma}_2}.
	\end{equation}	
	Suppose \eqref{eq:closep} is false and $\zalg > \zeta^\star$. Then letting $u \defeq p(\zalg)$ and $v \defeq p(\zeta^\star)$, note that $u$, $v$, and $u - v$ are all entrywise nonnegative and hence
	\[\norm{u}_2^2 \ge \norm{v}_2^2 + \sum_{i \in [n]} 2u_i(u_i - v_i) + (u_i - v_i)^2 > \norm{v}_2^2 + \Par{\frac{\delta}{4\norm{\gamma}_2}}^2. \]
	Hence, we have by $\sqrt{x^2 + y^2} \ge x + \frac{y^2}{3x}$ for $0 \le y \le x$, \eqref{eq:param-lowerbound}, and $\theta \le \norm{\gamma}_2$,
	\begin{align*}
	\norm{p(\zalg)}_2 = \norm{u}_2 > \norm{v}_2 + \frac{\Par{\frac{\delta}{4\norm{\gamma}_2}}^2}{3\norm{v}_2} \ge \theta + \frac{\delta^2}{48\norm{\gamma}_2^3} \ge \theta + \Delta + 2\beta.
	\end{align*}
	So, by triangle inequality $\norm{\tp(\zalg)}_2 > \theta + \Delta + \beta$ and hence we reach a contradiction with \eqref{eq:tprange}.
	
	Similarly, suppose \eqref{eq:closep} is false and $\zalg < \zeta^\star$. Then for the same definitions of $u$, $v$, and using the inequality $\sqrt{x^2 - y^2} \le x - \frac{y^2}{3x}$ for $0 \le y \le x$, we conclude
	\[\norm{v}_2^2 > \norm{u}_2^2 + \Par{\frac{\delta}{4\norm{\gamma}_2}}^2 \implies \norm{u}_2 \le \sqrt{\norm{v}_2^2 - \Par{\frac{\delta}{4\norm{\gamma}_2}}^2} < \theta - \Delta - 2\beta. \]
	So we reach a contradiction with \eqref{eq:tprange} in this case as well.
	
	In conclusion, \eqref{eq:closep} is true and we obtain by triangle inequality the desired
	\[\norm{\tp(\zalg) - p(\zeta^\star)}_2 \le \frac{\delta}{4\norm{\gamma}_2} + \beta\sqrt{d} \le \frac{\delta}{2\norm{\gamma}_2}. \]
	The complexity of the algorithm is bottlenecked by the cost of finding $\tp(\zalg)$. For each $\zeta$ on the grid the cost of evaluating $\tp(\zeta)$ induces a multiplicative $d\log(\frac{\norm{\gamma}_2^2}{\delta})$ overhead. The cost of performing the binary search on the $\zeta$ grid is a multiplicative $\log(\frac{\norm{\gamma}_2^2}{\mu\sqrt{\delta}})$ overhead; note that a binary search suffices because $\norm{\tp(\zalg)}_2$ is monotonic by our consistent choice of rounding down, and hence $|\norm{\tp(\zalg)}_2 - \theta|$ is unimodal.
\end{proof} 	\end{appendix}

\end{document}